\documentclass[11pt]{article}
\usepackage[  textheight=9in, textwidth=6.5in, top=1in, headheight=14pt, headsep=25pt, footskip=30pt]{geometry}
\usepackage{amsmath,amssymb,amsfonts,mathtools,amsthm}
\usepackage{bbm}
\usepackage{color,graphicx}
\usepackage{verbatim}
\usepackage[algo2e,linesnumbered,ruled,vlined]{algorithm2e}
\usepackage{algorithmic}
\usepackage{setspace}
\usepackage{caption}
\usepackage{subcaption}
\usepackage{fancyhdr}
\usepackage[colorinlistoftodos]{todonotes}
\usepackage{rotating}
\usepackage{booktabs}
\usepackage{enumitem, comment}
\usepackage{soul} 
\usepackage{dsfont}
\usepackage{bookmark}
\usepackage[utf8]{inputenc} 
\usepackage{csquotes}
\usepackage{fnpct} 
\usepackage{float}
\usepackage[utf8]{inputenc}
\usepackage[english]{babel}
\usepackage{parskip}
\usepackage{cleveref}
\usepackage[T1]{fontenc}
\usepackage{lmodern}
\usepackage{xspace}   

\usepackage{amsthm}
\usepackage{thm-restate}
\usepackage{ifdraft}
\usepackage{nicefrac}

\usepackage{mdframed}
\newmdtheoremenv{theomd}{Theorem}

\newtheorem{theorem}{Theorem}[section]
\newtheorem*{theorem*}{Theorem}

\newtheorem*{claim*}{Claim}

\newtheorem*{proposition*}{Proposition}
\newtheorem{lemma}[theorem]{Lemma}
\newtheorem*{lemma*}{Lemma}
\newtheorem{corollary}[theorem]{Corollary}

\newtheorem*{conjecture*}{Conjecture}

\newtheorem{fact}[theorem]{Fact}

\newtheorem*{hypothesis*}{Hypothesis}
\theoremstyle{definition}
\newtheorem{definition}[theorem]{Definition}
\newmdtheoremenv{defmd}[theorem]{Definition}

\newtheorem{problem}[theorem]{Problem}

\usepackage{prettyref}
\newcommand{\savehyperref}[2]{\texorpdfstring{\hyperref[#1]{#2}}{#2}}

\newrefformat{eq}{\savehyperref{#1}{\textup{(\ref*{#1})}}}
\newrefformat{lem}{\savehyperref{#1}{Lemma~\ref*{#1}}}
\newrefformat{def}{\savehyperref{#1}{Definition~\ref*{#1}}}
\newrefformat{thm}{\savehyperref{#1}{Theorem~\ref*{#1}}}
\newrefformat{cor}{\savehyperref{#1}{Corollary~\ref*{#1}}}
\newrefformat{cha}{\savehyperref{#1}{Chapter~\ref*{#1}}}
\newrefformat{sec}{\savehyperref{#1}{Section~\ref*{#1}}}
\newrefformat{app}{\savehyperref{#1}{Appendix~\ref*{#1}}}
\newrefformat{tab}{\savehyperref{#1}{Table~\ref*{#1}}}
\newrefformat{fig}{\savehyperref{#1}{Figure~\ref*{#1}}}
\newrefformat{hyp}{\savehyperref{#1}{Hypothesis~\ref*{#1}}}
\newrefformat{alg}{\savehyperref{#1}{Algorithm~\ref*{#1}}}
\newrefformat{sdp}{\savehyperref{#1}{SDP~\ref*{#1}}}
\newrefformat{rmk}{\savehyperref{#1}{Remark~\ref*{#1}}}
\newrefformat{item}{\savehyperref{#1}{Item~\ref*{#1}}}
\newrefformat{step}{\savehyperref{#1}{step~\ref*{#1}}}
\newrefformat{conj}{\savehyperref{#1}{Conjecture~\ref*{#1}}}
\newrefformat{fact}{\savehyperref{#1}{Fact~\ref*{#1}}}
\newrefformat{prop}{\savehyperref{#1}{Proposition~\ref*{#1}}}
\newrefformat{claim}{\savehyperref{#1}{Claim~\ref*{#1}}}
\newrefformat{relax}{\savehyperref{#1}{Relaxation~\ref*{#1}}}
\newrefformat{red}{\savehyperref{#1}{Reduction~\ref*{#1}}}
\newrefformat{part}{\savehyperref{#1}{Part~\ref*{#1}}}
\newrefformat{prob}{\savehyperref{#1}{Problem~\ref*{#1}}}
\newrefformat{ass}{\savehyperref{#1}{Assumption~\ref*{#1}}}
\newrefformat{ques}{\savehyperref{#1}{Question~\ref*{#1}}}

\newrefformat{ex}{\savehyperref{#1}{Example~\ref*{#1}}}

\usepackage{bm}

\hypersetup{colorlinks,linkcolor=blue,filecolor = blue, citecolor = violet, urlcolor  = magenta}

\usepackage[backend=biber, style=alphabetic,
sorting=nyt, natbib=true, backref=false, eprint=false, url=false, doi = false, isbn = false, maxbibnames = 99, maxcitenames = 2, maxalphanames = 4]{biblatex}
\DeclareSourcemap{
  \maps[datatype=bibtex]{
    \map[overwrite]{
      \step[fieldsource=eprint, final]
      \step[fieldset=doi, null]
    }
  }
}

\newcommand\blfootnote[1]{
  \begingroup
  \renewcommand\thefootnote{}\footnote{#1}
  \addtocounter{footnote}{-1}
  \endgroup
}

\DeclareSourcemap{
  \maps[datatype=bibtex]{
    \map[overwrite]{
      \step[fieldsource=doi, final]
      \step[fieldset=url, null]
    }
  }
}

\DeclareSourcemap{
  \maps[datatype=bibtex]{
    \map[overwrite]{
      \step[fieldsource=eprint, final]
      \step[fieldset=url, null]
    }
  }
}

\DeclareSourcemap{
  \maps[datatype=bibtex, overwrite]{
    \map{
      \step[fieldset=address, null]
      \step[fieldset=location, null]
    }
  }
}

\renewenvironment{abstract}
  {{\centering\large\bfseries Abstract\par}\vspace{0.7ex}%
    \bgroup
       \leftskip 20pt\rightskip 20pt\small\noindent\ignorespaces}%
  {\par\egroup\vskip 0.25ex}

\newlength\aftertitskip     \newlength\beforetitskip
\newlength\interauthorskip  \newlength\aftermaketitskip

\newcommand{\mper}{\,.}
\newcommand{\mcom}{\,,}

\newcommand{\paren}[1]{\left(#1 \right )}

\newcommand{\Brac}[1]{\left[#1\right]}

\newcommand{\set}[1]{\left\{#1\right\}}
\newcommand{\Set}[1]{\left\{#1\right\}}

\newcommand{\abs}[1]{\left\lvert#1\right\rvert}

\newcommand{\ceil}[1]{\lceil #1 \rceil}
\newcommand{\floor}[1]{\lfloor #1 \rfloor}

\newcommand{\norm}[1]{\left\lVert#1\right\rVert}

\newcommand{\fnorm}[1]{\norm{#1}_\mathrm{F}}

\newcommand{\spnorm}[1]{\norm{#1}_\mathrm{2}}
\newcommand{\nuclear}[1]{\norm{#1}_{*}}

\newcommand{\nnz}[1]{\mathrm{nnz}(#1)}

\newcommand{\defeq}{\stackrel{\textup{def}}{=}}

\newcommand{\inprod}[1]{\left\langle #1\right\rangle}

\newcommand{\Tr}[1]{\mathrm {Tr}\paren{#1}}
\newcommand{\tr}{\mathrm {Tr}}

\newcommand{\Z}{{\mathbb Z}}
\newcommand{\N}{{\mathbb Z}_{\geq 0}}
\newcommand{\R}{\mathbb R}

\newcommand{\Esymb}{\mathbb{E}}
\newcommand{\Psymb}{\mathbb{P}}

\DeclareMathOperator*{\E}{\Esymb}

\DeclareMathOperator*{\ProbOp}{\Psymb}

\newcommand{\given}{\mathrel{}\middle|\mathrel{}}

\newcommand{\prob}[1]{\ProbOp\Set{#1}}

\newcommand{\probSub}[2]{\mathbb{P}_{#2}\Set{#1}}

\renewcommand{\Pr}[1]{\ProbOp\Brac{#1}}

\newcommand{\e}{\epsilon}
\newcommand{\ebf}{\mathbf{e}}

\newcommand{\zero}{\mathbf{0}}

\definecolor{DSgray}{cmyk}{0,0,0,0.7}

\let\e\varepsilon

\newcommand{\cA}{\mathcal A}

\newcommand{\cD}{\mathcal D}

\newcommand{\cG}{\mathcal G}
\newcommand{\cH}{\mathcal H}
\newcommand{\cI}{\mathcal I}

\newcommand{\cS}{\mathcal S}
\newcommand{\cT}{\mathcal T}

\newcommand{\bbS}{\mathbb S}

\newcommand{\Erdos}{Erd\H{o}s\xspace}
\newcommand{\Renyi}{R\'enyi\xspace}

\newcommand{\bigO}{O}
\newcommand{\bigo}[1]{\bigO\!\left(#1\right)}

\newcommand{\poly}{{\sf poly}}

\newcommand{\supp}{{\sf supp}}

\newcommand{\tildeN}{\widetilde{N}}

\newcommand{\tX}{\widetilde{X}}
\newcommand{\tA}{\widetilde{A}}

\newcommand{\calI}{\mathcal{I}}

\usepackage{dsfont}
\usepackage{bbm}

\newcommand{\Ebar}{\bar{E}}

\newcommand{\spectralApx}{additive spectral norm approximation }

\newcommand{\nuclearApx}{nuclear norm approximation }

\usepackage{subcaption}

\renewcommand{\epsilon}{\varepsilon}

\newcommand{\Alg}{{\mathsf{ALG}}}

  \usepackage{nth}
  \usepackage{intcalc}

  \newcommand{\cAAAI}[1]{AAAI\ Conference\ on\ Artificial (AAAI)}

\title{Faster Spectral Density Estimation and \\ Sparsification in the Nuclear Norm}
\author{Yujia Jin \\ Stanford University \\ \texttt{yujia@stanford.edu} \and Ishani Karmarkar \\ Stanford University \\ \texttt{ishanik@stanford.edu} \and Christopher Musco \\ New York University \\ \texttt{cmusco@nyu.edu} \and Aaron Sidford \\ Stanford University \\ \texttt{sidford@stanford.edu} \and  Apoorv Vikram Singh \\ New York University \\ \texttt{apoorv.singh@nyu.edu}}

\date{}
\begin{document}
\maketitle

\begin{abstract}%
We consider the problem of estimating the spectral density of the normalized adjacency matrix of an $n$-node undirected graph. We provide a randomized algorithm that, with $O(n\epsilon^{-2})$ queries to a degree and neighbor oracle and in $O(n\epsilon^{-3})$ time, estimates the spectrum up to $\epsilon$ accuracy in the Wasserstein-1 metric. This improves on previous state-of-the-art methods, including an $O(n\epsilon^{-7})$ time algorithm from [Braverman et al., STOC 2022] and, for sufficiently small $\epsilon$, a $2^{O(\epsilon^{-1})}$ time method from [Cohen-Steiner et al., KDD 2018]. To achieve this result, we introduce a new notion of graph sparsification, which we call \emph{nuclear sparsification}. We provide an $O(n\epsilon^{-2})$-query and $O(n\epsilon^{-2})$-time algorithm for computing $O(n\epsilon^{-2})$-sparse nuclear sparsifiers. We show that this bound is optimal in both its sparsity and query complexity, and we separate our results from the related notion of additive spectral sparsification. Of independent interest, we show that our sparsification method also yields the first \emph{deterministic} algorithm for spectral density estimation that scales linearly with $n$ (sublinear in the representation size of the graph). 
\end{abstract}

\blfootnote{Accepted for presentation at the
Conference on Learning Theory (COLT) 2024}

\thispagestyle{empty}

\newpage


\pagenumbering{arabic} 
\section{Introduction}\label{sec:intro}
We study the fundamental problem of estimating the spectrum of an (undirected) graph. This problem has received significant attention due to applications in visualizing, classifying, and understanding large networks \citep{FarkasDerenyiBarabasiVicsek:2006, EikmeierGleich:2017,Cohen-SteinerKongSohler:2018,DongBensonBindel:2019,ChenTrogdonUbaru:2021,BravermanKrishnanMusco:2022}. Concretely, we consider a standard version of the problem based on the graph's normalized adjacency matrix:

\begin{problem}[Spectral Density Estimation (SDE)]
\label{prob:spectral_density}
Given an $n$-node undirected graph $G = (V, E, w)$ with positive edge weights $w \in \R^{E}_{> 0}$, return eigenvalue estimates $\widehat{\lambda}_1\le \cdots\le\widehat{\lambda}_n$ such that
\begin{align}
\label{eq:wasserstein-lambda}    
\frac{1}{n} \sum_{i\in\{1,\ldots, n\}}|\widehat{\lambda}_i-\lambda_i(N_G)|\le \e,
\end{align}
where ${\lambda}_1(N_G)\le \cdots\le{\lambda}_n(N_G)$ are the eigenvalues of $G$'s normalized adjacency matrix, $N_G \in \R^{V \times V}$.
\end{problem}

Note that \eqref{eq:wasserstein-lambda} is equivalent to requiring that the Wasserstein-1 distance between the uniform distribution on $\lambda_1, \ldots, \lambda_n$ and the uniform distribution on $\widehat{\lambda}_1, \ldots, \widehat{\lambda}_n$ is less than $\epsilon$.
We also consider a more general problem, which we call \emph{the matrix SDE problem}, where the goal is to estimate eigenvalues of a given symmetric matrix with spectral norm bounded by $1$, in the same metric, \eqref{eq:wasserstein-lambda}.

Outside of applications in network science, SDE has seen recent interest from the machine learning community due to its connections to learning distributions based on moment estimates \citep{KongValiant:2017,Jin23} and learning graph structure from random walk traces 
\citep{Cohen-SteinerKongSohler:2018}. The matrix SDE problem also has applications in deep learning \citep{GhorbaniKrishnanXiao:2019,PenningtonSchoenholzGanguli:2018}. Algorithms for related (often stronger) notions of spectral approximation, like relative-error eigenvalue histograms, or point-wise eigenvalue estimates, have been widely studied for applications like approximating matrix norms and spectral sums \citep{MuscoNetrapalli19,Bhattacharjee22,Swartworth23,BhattacharjeeDexterMusco:2024}.

Excitingly, at least for unweighted graphs, it has been shown that \Cref{prob:spectral_density} can be solved in time that is \emph{sublinear} in the size of $G$'s representation, which can be as large as $\Omega(n^2)$. Given a ``neighbor'' oracle that when queried with $a \in V$ and $i \in \Z_{> 0}$ outputs the degree of $a$ and the $i$-th vertex incident to $a$ (if there is one) in $\bigO(1)$ time, there are randomized SDE algorithms running in $2^{\bigO(\epsilon^{-1})}$ \citep{Cohen-SteinerKongSohler:2018} and $\bigO(n \epsilon^{-7})$ time \citep{BravermanKrishnanMusco:2022}.\footnote{The runtime of the method from  \cite{Cohen-SteinerKongSohler:2018} is independent of $n$; instead of outputting a list of $n$ eigenvalues, it outputs a list of $\bigO(1/\epsilon)$ distinct eigenvalue magnitudes and corresponding multiplicities.}
Recent work also studies query complexity lower bounds for different oracles \citep{Jin23}.
However, large gaps remain: The best query lower bound for neighbor oracles is just  $\Omega(1/\epsilon^2)$ \citep{Jin23}. 

Our work is motivated by these advances, the general importance of sublinear time graph algorithms \citep{chazelle2005approximating,czumaj2003sublinear,CzumajSohler:2010,onak2012near,berenbrink2014estimating,eden2018approximating,mishra2001sublinear}, and progress on sublinear time algorithms for other spectral problems like expander testing \citep{CzumajSohler:2007,NachmiasShapira:2010,GoldreichRon:2011} and spectral clustering \citep{gluch2021spectral,mousavifar2021sublinear}. 
We ask: Is it possible to improve on existing SDE algorithms? Is it possible to efficiently obtain new sparse approximations, i.e., \emph{sparsifications}, of $G$ that facilitate more efficient spectral density estimation? 
Finally, all previous sublinear time SDE algorithms make critical use of randomness. Are there sublinear time \emph{deterministic} algorithms?

In this paper, we provide an affirmative answer to each of these questions. 
For unweighted graphs, we provide a randomized $\bigO(n \epsilon^{-3})$-time algorithm for solving \Cref{prob:spectral_density} and a deterministic $n \cdot 2^{\tilde{O}(\epsilon^{-1})}$-time algorithm. The randomized method improves on the $\bigO(n \epsilon^{-7})$-time method from \citet{BravermanKrishnanMusco:2022} and on the $2^{\bigO(\epsilon^{-1})}$-time method from \citet{Cohen-SteinerKongSohler:2018} for small values of $\epsilon.$ To the best of our knowledge, our deterministic method is the first to achieve a sub-quadratic dependence on the number of nodes, $n$.
Moreover, we obtain an algorithm with the same complexity for weighted graphs provided that, when queried with $a \in V$ and $i \in \Z_{> 0}$, the oracle outputs the weighted-degree of $a$, the weight of the $i$-th largest edge incident to $a$ (with ties broken arbitrarily), and the endpoint of this edge. Finally, we show that these complexities are obtainable even in a weaker \emph{random walk model} for accessing $G$.

\subsection{A Sparsification Approach to Spectral Density Estimation}\label{intro:density_approach_1}
Our results follow a natural \emph{two-stage} approach to approximating the spectral density of graphs. First, in sublinear time, we construct a \emph{sparse matrix} that approximates $N_G \in \R^{V \times V}$, the normalized adjacency matrix of $G$, and has $o(n^2)$ entries. Second, we apply and adapt existing spectral density estimation methods to efficiently approximate the spectrum of the sparse approximation. This approach decouples the information theoretic problem of how to approximate $N_G$ using a small number of queries from the {computational} problem of how to efficiently approximate eigenvalues.

Concretely, our approach motivates the following natural question: What notions of graph approximation (ultimately, sparsification) are sufficient for preserving the spectrum in the sense of \eqref{eq:wasserstein-lambda} \emph{and} can be obtained in sublinear time?
We make progress on this problem by introducing a new notion of \emph{$\epsilon$-additive nuclear approximation} and presenting algorithms that obtain near-optimal sparsity and query complexity for producing such approximations. Formally, we define:

\begin{definition}[Additive Nuclear Approximation and Sparsification]\label{def:nuclear-sparsifier-informal} $M \in \R^{V \times V}$ is an $\e$-additive \emph{nuclear approximation} to $n$-vertex graph $G = (V,E,w)$ if $\nuclear{N_G-M}\le \e n$. $M$ is an $\e$-additive \emph{nuclear sparsifier} of $G$ if it is also $s$-sparse, i.e., has at most $s$ non-zero entries, for $s = o(n^2)$.\footnote{Note that $M$ itself is not required to be the adjacency matrix of a graph and our algorithms will in general return matrices that are not. However, in all cases, we show that one can modify our algorithm to ensure that the output is indeed a normalized adjacency matrix. See \Cref{lemma:graphical-nuclear-ub} in \Cref{app:omitted} for details.
}
\end{definition}
In \Cref{def:nuclear-sparsifier-informal}, $\nuclear{N_G-M}$ denotes the nuclear norm, i.e., the sum of the singular values of $N_G-M$. Importantly, additive nuclear approximation is sufficient for solving the SDE problem. For symmetric matrices, $A, B \in \R^{n \times n}$, let $W_1(A, B) \defeq \frac{1}{n}\sum_{i=1}^n|{\lambda}_i(A)-\lambda_i(B)|$ denote Wasserstein-$1$ distance between the probability distributions $p$ and $q$ induced by the real-valued eigenvalues of $A$ and of $B$. A short proof (see \Cref{sub:notation}, \Cref{lem:redx-sde}) establishes that $ W_1(A,B)\le \frac{1}{n}\nuclear{A-B}$.\footnote{When $A$ and $B$ are diagonal with monotonically decreasing diagonal entries then $W_1(A,B) = \frac{1}{n}\|A - B\|_*$. In this sense, nuclear approximation is a natural strengthening of approximation in the Wasserstein-1 distance.}
As a result, if $\|M - N_G\|_* \leq \epsilon n/2$ and we find eigenvalues $\widehat{\lambda}_1, \ldots, \widehat{\lambda}_n$ that are $\epsilon/2$ close in Wasserstein distance to those of $M$, then by triangle inequality, those eigenvalues solve \Cref{prob:spectral_density}. 

Our main algorithmic result is that nuclear sparsifiers with $O(n\epsilon^{-2})$ non-zero entries can be computed \emph{deterministically} in $O(n\epsilon^{-2})$ time. Formally, we assume the following access to $G$:
\begin{definition}[Adjacency query model]\label{def:query_model} We say we have \emph{adjacency query access} to a weighted graph $G=(V,E,w)$ if $V$ is known and there is an $O(1)$ time procedure, $\mathsf{GetNeighbor}(a,i)$, that when queried with any $a \in V$ and $i \in \Z_{> 0}$ outputs $\deg_G(a)$ and, if there is one, the $i$-th largest edge $\{a,b\}$ and $w_{\{a,b\}}$ (with ties broken arbitrarily). If the $a$ has $< i$ edges, the oracle returns $\emptyset$.
\end{definition}
Adjacency queries are easily supported by standard graph data structures. For example, an adjacency list where each node's neighbors are stored in an array supports adjacency queries for unweighted graphs. To support weighted graphs, it suffices to sort these arrays by edge weight. \footnote{While this assumption on the ordering of the array might be considered non-standard, it is easy to support using a basic data structure for storing the graph provided in which we sort the edges by weight. We also show in Theorem~\ref{thm:nuclear_sparsifier_rw} that we can extend our result to a more standard random walk model, still obtaining a nuclear sparsifier with just $O(n\epsilon^{-2})$ queries.}

Within the query model established, we prove the following main result on constructing sublinear size nuclear sparsifiers in \Cref{sec:upper}. Our method works even for weighted graphs.
\begin{restatable}[Sublinear Time Nuclear Sparsification]{theorem}{nuclearub}\label{thm:nuclear-ub-informal}
     There is a \emph{deterministic} method (Algorithm \ref{alg:nuclear}) that, for any $\e\in(0,1)$,  returns a $2n \e^{-2}$-sparse $\e$-additive nuclear sparsifier for any undirected weighted graph $G$, and runs in $\bigO(n \e^{-2})$ time in the adjacency query model (\Cref{def:query_model}).
\end{restatable}
We obtain \Cref{thm:nuclear-ub-informal} using a greedy approach that deterministically adds edges from $G$ to the sparsifier based on their weight and end-point degrees. Importantly, we leverage that this approach is deterministic to obtain the first \emph{deterministic sublinear time SDE algorithms}.
Before discussing these SDE algorithms, in \Cref{sec:sparsification_compare} we highlight that deterministic sublinear time methods like \Cref{thm:nuclear-ub-informal} are \emph{not possible} for stronger and more well-studied notions of graph sparsification.

\subsection{Comparison to Prior Work on Sparsification}
\label{sec:sparsification_compare}

A significant line of work studies \emph{spectral graph sparsification}, which generalizes cut sparsification \citep{BenczurKarger:1996}, and has applications in linear system solving, combinatorial graph algorithms, and beyond \citep{spielman2011spectral,SpielmanSrivastava:2011,BatsonSpielmanSrivastava:2012,KapralovLeeMusco:2017,lee2018constructing}. Concretely, a spectral sparsifier is defined as:

\begin{definition}[Spectral Sparsifier, \cite{spielman2011spectral}]\label{def:spectral-sparsifier} Given $\e > 0$ and a graph $G$ with (unnormalized) Laplacian matrix $L$, a symmetric matrix $\widetilde{L} \in \R^{n \times n}$ is an $\e$-spectral sparsifier of $L$ if, for all $x \in \R^n$, $(1-\e) x^\top L x \leq x^\top \widetilde{L} x \leq (1+\e) x^\top L x$. 
\end{definition}
\Cref{def:spectral-sparsifier} involves the Laplacian, whereas we focus on the (normalized) adjacency matrix. Nevertheless, it is straightforward to show that the spectrum of the normalized adjacency matrix of an $\e$-spectral sparsifier is $\e$-close to that of $G$ in Wasserstein distance.
However, it is impossible to obtain $\epsilon$-spectral sparsifiers with $o(n^2)$ adjacency queries (see Theorem 11 of \citet{lee2013probabilistic} with $\delta = 1/n$). Consequently, we consider notions of spectral sparsification that are weaker than \Cref{def:spectral-sparsifier}, but stronger than nuclear sparsification (\Cref{def:nuclear-sparsifier-informal}). In particular, we introduce the following notion of \emph{additive error} spectral sparsification, which strictly strengthens  \Cref{def:nuclear-sparsifier-informal}:\footnote{Definition~\ref{def:eps-additive-spectral-spars} is related to two notions of sparsification previously studied by \citet{lee2013probabilistic} and \citet{agarwal2022sublinear}. Additive spectral sparsification is stronger than probabilistic $(\epsilon, \epsilon)$-spectral sparsification \cite{lee2013probabilistic} because we require $x^\top \tilde{N}_G x - \epsilon \cdot x^\top x \leq x^\top M x \leq x^\top \tilde{N}_G x + \epsilon \cdot x^\top x$ to hold \emph{simultaneously for all vectors} $x \in \R^n$ with high probability. Additive spectral sparsification is also a stronger notion than $(\epsilon, \epsilon)$-cut sparsification \cite{agarwal2022sublinear}, because we require the additive guarantee to hold for all vectors $x$ rather than just for cuts.}
\begin{definition}[Additive Spectral Sparsifier]\label{def:eps-additive-spectral-spars} A symmetric matrix $M \in \R^{V \times V}$ is an $\epsilon$-additive spectral sparsifier of $G$ if $\|M - \widetilde{N}_G\|_2 \leq \epsilon$. 
\end{definition}

We show in \Cref{sec:rw} that an $O(n\epsilon^{-2}\log n )$-sparse $\epsilon$-additive spectral sparsifier can also be obtained for all weighted graphs in $O(n\epsilon^{-2}\log n )$ time using standard random sampling methods.\footnote{
The $O(n\epsilon^{-2}\log n )$ sparsity can be improved to $O(n\epsilon^{-2})$ at the cost of additional $\poly(\log n/\epsilon)$ runtime factors using near-linear time $\epsilon$-spectral sparsification methods \citep{kevin_bss, lee2018constructing}. }
Given this result, 
it is natural to ask whether we can obtain a \emph{deterministic} algorithm for additive spectral approximation with a linear dependence on $n$, as we do for nuclear sparsification in \Cref{thm:nuclear-ub-informal}.
Interestingly, we prove that this is impossible: 
\begin{restatable}{theorem}{addspecthm}\label{lem:norm_query_spectral-informal}
    Any deterministic algorithm requires $\Omega(n^2)$ Adjacency Queries (\Cref{def:query_model}) to compute a $\frac{1}{4}$-additive spectral sparsifier (\Cref{def:eps-additive-spectral-spars}), even for unweighted graphs.
\end{restatable}
\Cref{lem:norm_query_spectral-informal} highlights a strong separation between our new notation of nuclear sparsification, for which we have a sublinear time deterministic algorithm, and stronger notions of sparsification. 

\subsection{Applications to Sublinear Time Spectral Density Estimation}
A main application of our result on nuclear sparsification (\Cref{thm:nuclear-ub-informal}) is faster deterministic and randomized sublinear time SDE algorithms for graphs. In the randomized setting, a long line of work in computational chemistry, applied math, and, recently, computer science \citep{Skilling:1989,SilverRoder:1994,WeisseWelleinAlvermann:2006,LinSaadYang:2016} studies \emph{linear time} methods for spectral density estimation. In \citet{ChenTrogdonUbaru:2021} and \citet{BravermanKrishnanMusco:2022} it was proven that common randomized algorithms like the stochastic Lanczos quadrature method and moment matching can approximate the spectral density of \emph{any} symmetric matrix $A$ up to $\epsilon$-accuracy in Wasserstein distance using roughly $\bigO(\epsilon^{-1})$ matrix-vector multiplications with $A$.\footnote{The method from \citet{BravermanKrishnanMusco:2022} uses $O(\min(\epsilon^{-1}, \epsilon^{-2}\log^4(\epsilon^{-1})/n))$ matrix-vector multiplications, which is $\bigO(\epsilon^{-1})$ for sufficiently large $n$.}
At a high level, these methods solve the SDE problem by approximating the first $q = O(\epsilon^{-1})$ \emph{moments} of $A$'s eigenvalue density, which can be expressed as $\tr(A), \tr(A^2), \ldots, \tr(A^q)$. They return a distribution that matches or nearly matches those moments, so will be increasingly close to the true spectral density as $q$ increases. For each $i\in 1, \ldots, q$, estimating $\tr(A^i)$ can be reduced to roughly $O(i)$ matrix-vector products with $A$ through the use of stochastic trace estimation methods like Hutchinson's estimator \cite{Hutchinson:1990,MeyerMuscoMusco:2021}.

We obtain our main result on the SDE problem by simply applying such methods (concretely, Theorem 1.4 from \citet{BravermanKrishnanMusco:2022}) to our nuclear norm sparsifier, which can be multiplied by a vector in $O(n\epsilon^{-2})$ time. Formally, we obtain the following result (proven in \Cref{sec:upper}):
\begin{restatable}[Randomized Sublinear Time SDE]{theorem}{thmsde}\label{corr:sde-informal}
    There is a randomized algorithm that solves the SDE problem (Problem \ref{prob:spectral_density})  with probability $99/100$ in $\bigo{n\e^{-3}}$ time in the adjacency query model (\Cref{def:query_model}).
\end{restatable}

\Cref{corr:sde-informal} directly improves on an $O(n\epsilon^{-7})$ time method from \citet{BravermanKrishnanMusco:2022}, which only applied to unweighted graphs, and takes a different approach from ours. That work leverages randomized methods for approximating matrix-vector products instead of sparsification.
For sufficiently small $\epsilon$, we also improve on the $2^{O(\e^{-1})}$ time method from  \citet{Cohen-SteinerKongSohler:2018}, which has {no dependence} on $n$.\footnote{Initial evidence (lower bounds in restricted models) suggests that it may not be possible to improve the $\e$ dependence in \citet{Cohen-SteinerKongSohler:2018} to sub-exponential while maintaining no dependence on $n$ \citep{Jin23}.} 

In addition to a faster randomized algorithm, in \Cref{sec:upper} we show how to use our deterministic nuclear norm sparsifiers to obtain the first sublinear time deterministic SDE method for graphs: 
\begin{restatable}[Deterministic Sublinear Time SDE]{theorem}{thmsdedeter}
\label{corr:sde-deterministic}
    There is a deterministic algorithm that solves the SDE problem (Problem \ref{prob:spectral_density}) in $n \cdot 2^{O(\e^{-1}\log(\e^{-1})}$ time in the adjacency query model (\Cref{def:query_model}).
\end{restatable}
We prove
\Cref{corr:sde-deterministic} by showing that it is possible to exactly compute the $i^\text{th}$ eigenvalue moment, $\Tr{M^i}$, for our particular nuclear norm sparsifier $M$ in $O(2^{i\log(\epsilon^{-1})}))$ time via direct computation of the diagonal entries of $\Tr{M^i}$. Crucially, we leverage the fact that our nuclear norm sparsifiers guaranteed by \Cref{thm:nuclear-ub-informal} are \emph{uniformly} sparse: not only is the total number of non-zero entries in $M$ small, but every \emph{row} of the sparsifier has a bounded number of non-zero entries (specifically, $O(1/\epsilon^2)$). Once the eigenvalue moments are computed, we again appeal to the moment matching method from \citet{BravermanKrishnanMusco:2022}, as in the proof of \Cref{corr:sde-informal}.
Previously, the best deterministic algorithm for spectral density estimation was to perform an eigendecomposition of $N_G$ in time $\bigo{n^\omega}$, where $\omega \leq 2.371552$ is the current fast matrix multiplication exponent \citep{williams2024new}.

\subsection{Lower Bounds for Nuclear Sparsification}\label{sec:optimality}
Given that nuclear approximation is a natural relaxation of widely studied graph sparsification notions like spectral sparsification, it is desirable to fully understand the complexity of the problem. We complement our algorithmic result from \Cref{thm:nuclear-ub-informal} with nearly matching sparsity and query lower bounds. First, we show that the sparsity of our sparsifiers is near-optimal:

\begin{restatable}[Sparsity Lower Bound]{theorem}{sparsitynuclear}\label{lem:norm_adj-informal}
    For any $\e \leq \e_0$, where $\e_0\in(0,1)$ is a fixed constant and any integer $n \geq 1/\e^2$, there is a graph $G$ on $n$ nodes, with normalized adjacency matrix $N_G$, such that  any matrix $M$ satisfying $\nuclear{N_G - M} \leq n \e$ must have  $\Omega(n\e^{-2}/\log^2\e^{-1})$ non-zero entries.
\end{restatable}
We prove \Cref{lem:norm_adj-informal} by considering the extreme case when $\epsilon = 1/\sqrt{cn}$ for a small constant $c$. This is the smallest value of $\epsilon$ for which \Cref{thm:nuclear-ub-informal} gives a non-trivial result (i.e., a matrix with less than $n^2$ entries). Using the probabilistic method, we show that there are $2^{O(n^2)}$ graphs whose adjacency matrices are all $\epsilon$-far in the nuclear norm. A pigeonhole argument is then used to show that not all of these matrices can be well approximated by $O(n^2/\log^2 n)$-sparse matrices. We prove the result for general $\epsilon$ via a reduction to the $\epsilon = 1/\sqrt{cn}$ case.

We also show that any algorithm for constructing an $\epsilon$-additive nuclear sparsifier must, in the worst case, make $\widetilde{\Omega}(n\epsilon^{-2})$ $\mathsf{GetNeighbor}$ queries, matching \Cref{thm:nuclear-ub-informal} up to log factors:

\begin{restatable}[Query Lower Bound]{theorem}{querynuclearlb}\label{lem:norm_query-informal}
  For any $\e \leq \e_0$ and any integer $n \geq C/\e^2$ where $\e_0\in(0,1)$ and $C> 1$ are fixed constants, any algorithm that returns an $\e$-additive nuclear sparsifier for any input $G$ with probability $\geq 3/4$ requires $\Omega(n\e^{-2} \log^{-2}(\e^{-1}))$ $\mathsf{GetNeighbor}$ queries.
\end{restatable}
Notably, \Cref{lem:norm_query-informal} even applies to randomized algorithms. 

It is interesting to ask if the same query complexity is optimal for spectral density estimation itself (\Cref{prob:spectral_density}). Currently, the best query lower bound, due to a recent result of \citet{Jin23},  is just $\Omega(1/\epsilon^2)$. That work also proves a lower bound of $\Omega(2^{1/\e})$ under a more restrictive random walk query model (see \Cref{sec:intro:random-walk}).  However, large gaps still remain in understanding the optimal query complexity and running times for spectral density estimation. Addressing these gaps is an exciting direction for future work.

\subsection{Random Walk Query Model}
\label{sec:intro:random-walk}
Finally, motivated by the SDE algorithms of \citet{BravermanKrishnanMusco:2022} and \citet{Cohen-SteinerKongSohler:2018}, we consider a weaker graph access model than the adjacency query model. The algorithms in \citet{BravermanKrishnanMusco:2022} and \citet{Cohen-SteinerKongSohler:2018} assume access to $G$ via random walks, i.e., that in $O(k)$ time we can sample a random walk $v_0,v_1,...,v_k$, where $v_0$ is a uniformly random node in $G$ and $v_i$ is chosen from the neighbors of $v_{i-1}$ with probability proportional to edge weight. 

Interestingly, we show that it is possible to construct nuclear norm sparsifiers in an even weaker model where only one-step walks are allowed:
 
\begin{definition}[One-step Random Walk Query Model]\label{def:rw_query} We say we have \emph{one-step random walk query access} to a graph $G=(V, E, w)$ if there is an $O(1)$ time procedure, $\mathsf{RandomNeighbor}$, that returns a uniformly random vertex $a \in V$ and $\emptyset$ if $a$ has no neighbors, or an edge $\{a,b\}$ selected with probability proportional to its weight, along with the degree of $a$ and $b$.
\end{definition} 

While our algorithm discussed in \Cref{intro:density_approach_1} is not directly implementable in this more restrictive query model, in \Cref{sec:rw}, we present an alternative, randomized algorithm that achieves identical query and time-complexities in the model:
\begin{restatable}{theorem}{thmrwnuclearsparsifier}\label{thm:nuclear_sparsifier_rw} 
There is an algorithm (Algorithm~\ref{alg:sparsify-nuclear}) that, for any $\e \in (0, 1)$, returns with probability $2/3$ an $O(n \e^{-2}$)-sparse $\e$-additive nuclear sparsifier for any undirected weighted graph $G$ using $O(n\epsilon^{-2})$ queries in the one-step random walk query model (Definition~\ref{def:rw_query}).
\end{restatable}

We also obtain an algorithm for $\e$-additive spectral sparsification in the one-step random walk query model. We obtain the following via a natural application of matrix concentration. Specifically, we directly use an algorithm and theorem from \citet{cohen2017almost}. 
\begin{restatable}{theorem}{thmrwspectralsparsifier}\label{thm:spectral_sparsifier_rw}
There is an algorithm (Algorithm~\ref{alg:sparsify-spectral}) that, for any $\e \in (0, 1)$, returns with probability $2/3$ an $O(n \e^{-2} \log n$)-sparse $\e$-additive spectral sparsifier for any undirected weighted graph $G$ using just  $O(n\epsilon^{-2}\log n)$ queries in the one-step random walk query model (Definition~\ref{def:rw_query}).
\end{restatable}
Theorem~\ref{thm:spectral_sparsifier_rw} achieves the stronger notion of $\epsilon$-additive spectral sparsification, at the cost of an extra factor of $\log n $ in the sparsity, query complexity, and runtime as compared to Theorem~\ref{thm:nuclear_sparsifier_rw}.
We show that this $\log n$ factor is unavoidable in the following sense: 

\begin{restatable}{theorem}{lbspectraladdcoupon} \label{lemma:lb_spectral_add_coupon}
  For a fixed constant $\e \in (0, 1/8)$ and $c \in (0, 1)$ any algorithm requires $\Omega(n \log n)$ one-step random walk model queries to output an $\epsilon$-additive spectral sparsifier with probability $c$.
\end{restatable}
We prove \Cref{lemma:lb_spectral_add_coupon} via a coupon collector argument. Consider a graph of size $2n$ with $n$ isolated pairs of nodes, each of which is either connected by an edge or not. To obtain a constant factor additive spectral sparsifier, we must observe one node in each pair, which takes $\Omega(n\log n)$ samples via the standard coupon collector lower bound. Note that this 
lower bound also applies to the stronger $k$-step random walk model studied in \citet{BravermanKrishnanMusco:2022} and \citet{Cohen-SteinerKongSohler:2018}, since for such a graph, no additional information is gained from a longer walk.

\subsection{Paper Organization and Preliminaries}
\label{sub:notation}

\paragraph{Paper Organization.} \Cref{sec:upper} presents our nuclear sparsification and SDE algorithms in the adjacency query model. \Cref{sec:lower} and \Cref{sec:lower-query} present our sparsity and query lower bounds for nuclear sparsification. \Cref{sec:additive-spectral} presents a lower bound against deterministic algorithms for $\epsilon$-additive spectral sparsification. \Cref{sec:rw} covers our results in the more restrictive one-step random walk model. \Cref{app:omitted} shows how our algorithms for constructing a nuclear sparsifier can be modified to obtain a \emph{graphical} nuclear sparsifier (i.e., a matrix that is a nuclear sparsifier and is also the normalized adjacency matrix of some graph.)
\paragraph{Graph Notation.} 
In this paper, we consider undirected graphs $G=(V,E,w)$ with positive edge weights $w\in \R^E$. $A_G\in \R^{V\times V}_{\geq 0}$ denotes the adjacency matrix of $G$, i.e., $[{A_G}]_{v,v'} \defeq w_e$ if $ e = \{v,v'\}\in E$ and is $0$ otherwise. $\deg(v) \defeq \sum_{e = \{v,v'\} \in E} w_e$ denotes the weighted degree of vertex $v$, and $D_G \in \R^{V\times V}_{\geq 0}$ is the diagonal degree matrix of $G$, where ${D_G}_{v,v} := \deg(v)$ for all $v \in V$. Throughout, we let $N_G \defeq D_G^{-1/2} A_G D_G^{-1/2}$ denote the normalized adjacency matrix of $G$. 

\paragraph{Vector and Matrix Notation.} For positive integers $n$, we let $[n] \defeq \set{1,\dots,n}$. For a vector $v\in\R^n$, $\norm{v} \defeq \sqrt{\sum_{i\in[n]}v_i^2}$ denotes its Euclidean norm. We let $\mathbb{S}^{n \times n}$ denote the set of all real symmetric $n \times n$ matrices. 
The spectral norm of a matrix $A \in \R^{n \times n}$ is $\spnorm{A} \defeq \max_{v:v\in\R^n}{\norm{Av}}/{\norm{v}}$, which equals $\max_{i}|\lambda_i(A)|$ when $A\in \mathbb{S}^{n \times n}$. The nuclear norm of a matrix $A\in \mathbb{S}^{n \times n}$ is $\nuclear{A} \defeq \sum_{i=1}^{n} |\lambda_i(A)|$ and the Frobenius norm is $\fnorm{A} \defeq \sqrt{\sum_{i,j}A(i,j)^2} \sum_{i=1}^{n} |\lambda_i(A)|^2$. For matrices $A,X \in \R^{m\times n}$ we define the matrix inner-product,  $\inprod{A,X} \defeq \sum_{i,j}A(i,j)X(i,j)$.  
We use Loewner order notation $A\succcurlyeq 0$ or $0\preccurlyeq A$ to denote that a symmetric $A$ is positive semidefinite (PSD), i.e., that $A$ has non-negative eigenvalues. $A\succcurlyeq B$ denotes that $A-B$ is PSD. 

\paragraph{Other Notation.} We let $\mathrm{Ber}(p)$ denote a Bernoulli random variable with parameter $p$, i.e.,  $\mathrm{Ber}(p) = 1$ with probability $p$ and $0$ otherwise. $\mathcal{G}(n, p)$ denotes the \Erdos-\Renyi model with parameter $p$, i.e., a graph distributed as $\mathcal{G}(n, p)$ has $(i,j)\in E$ with probability $p$ for every $i<j\in [n]$, independently at random.

\paragraph{Preliminary Facts.} We rely on the following basic facts about the nuclear norm.

\begin{fact}[Nuclear Norm Duality]\label{fact:norm-ineqs}
 For any $A \in \mathbb{S}^{n \times n}$, 
$\nuclear{A} 
= \max_{Y \in \mathbb{S}^{n \times n}:\spnorm{Y}= 1}\inprod{A,Y}$.
 \end{fact}

\begin{fact}\label{lem:redx-sde} 
For $A,B\in\mathbb{S}^{n\times n}$, if $\nuclear{A-B}\le n\e$, then $ W_1(A,B)\le \e$.
\end{fact}
 \begin{proof}
 Let $\lambda_{\min}$ be the smallest eigenvalue among the eigenvalues of $A$ and $B$, so that $A - \lambda_{\min}I$ and $B - \lambda_{\min}I$ are both positive semidefinite, where $I$ denotes an $n\times n$ identity matrix. The singular values of $A - \lambda_{\min}I$ and $B - \lambda_{\min}I$ equal the eigenvalues of those matrices. So, we can directly apply Mirsky's singular value perturbation inequality \citep[Theorem 5]{Mirsky} to conclude that:
\begin{align*}
    \nuclear{A-B} = \nuclear{A-\lambda_{\min}I-(B-\lambda_{\min}I)} &\geq \sum_{i\in[n]}|\lambda_i(A-\lambda_{\min}I)-\lambda_i(B-\lambda_{\min}I)|\\
    &= \sum_{i\in[n]} \left|\lambda_i(A)-\lambda_i(B)\right| = nW_1(A,B). \qedhere
\end{align*}
\end{proof}

\section{Additive Nuclear Sparsifiers}\label{sec:upper}
In this section, we present our main algorithm (\Cref{alg:nuclear}) for constructing $O(n\e^{-2})$-sparse $\epsilon$-additive nuclear norm sparsifiers (\Cref{thm:nuclear-ub-informal}). 
We first prove a structural result, which leads to our greedy sparsification procedure. Specifically, we show that the normalized adjacency matrix can be sparsified by eliminating edges with small weights compared to the degree of its endpoints.

\begin{theorem}\label{thm:nuclear_ub}
  Let $G = (V, E, w)$ be a weighted graph with adjacency matrix $A_G \in \R^{n \times n}$, degree matrix $D_G$, and normalized adjacency matrix $N_G$. For $\e\in(0,1)$ define $E' \subseteq E$ as:
  \[  E' \defeq \Set{ e=\set{v,v'}\in E \given w_e \geq \frac{\e^2}{2} \cdot \max\set{\deg(v), \deg(v')} } \mper   \]
  Let $G'$ be obtained by removing all edges from  $G$ that are not in $E'$,  and let $A_{G'} \in \R^{n \times n}$ denote the adjacency matrix of $G'$. Let $\tildeN \defeq D_G^{-1/2} A_{G'} D_G^{-1/2} $. Then $\|N_G-\tildeN\|_* \leq \e n$ and $\tildeN$ has at most $2\e^{-2}$ non-zeros in each row (resp. column) and $\|\tildeN\|_2 \leq 1$.
\end{theorem}

\begin{proof}
  Let $D_{G'}$ denote the degree matrix of $A_{G'}$. Since $D_{G'}$ is entrywise smaller than $D_G$, and since $D_{G'} - A_{G'}$ is a graph Laplacian and thus PSD, we have the relationship:
  \[ -D_G \preccurlyeq -D_{G'} \preccurlyeq A_{G'}  \preccurlyeq D_{G'} \preccurlyeq D_G. \]
  Our claim that $\|\tildeN\|_2 \leq 1$ then follows by multiplying by $D_G^{-1/2}$ on the left and right:
    \begin{align*} 
    -I = -D_G^{-1/2}D_GD_G^{-1/2} \preccurlyeq D_G^{-1/2}A_{G'}D_G^{-1/2} = \tildeN \preccurlyeq D_G^{-1/2}D_GD_G^{-1/2} = I.
    \end{align*}
  Next, observe that for each $v \in V$ there are at most $2\e^{-2}$ nodes $v'$ for which $w_{\set{v,v'}}\ge \frac{\e^2}{{2}} \deg(v)$. Only edges for which $w_{\set{v,v'}}\ge \frac{\e^2}{{2}} \deg(v)$ are contained in $E$. Consequently, $A_{G'}$ and $\tildeN$ have at most $2\e^{-2}$ non-zeros per row/column, as claimed
  
  Finally, we bound the error in the nuclear norm by first bounding the Frobenius norm error. 
  Without loss of generality, assume that $V = \{1, \ldots, n\}$ and the vertices are ordered such that $\deg(i) \geq \deg(j)$ for $i < j$.
  \begin{align*}
    \norm{N_G-\tildeN}_{F}^2 = 2\sum_{\{i,j\} \in E \setminus E'} \paren{\frac{w_{\{i,j\}}}{\sqrt{\deg(i)\deg(j)}}}^2  
    &= 2 \sum_{i} ~ \sum_{j: \set{i,j} \in E \setminus E',\, j < i} \frac{w_{\{i,j\}}^2}{\deg(i)\deg(j)} \\
    & \leq 2 \sum_{i} ~ \sum_{j: \set{i,j} \in E \setminus E'  ,\, j < i} \frac{w_{\{i,j\}}}{\deg(i)} \frac{\e^2}{2}  \\
    &= \e^{2} \sum_{i} \frac{1}{\deg(i)} \sum_{j: \set{i,j} \in E \setminus E',\, j < i} w_{\{i,j\}} \leq \e^2 n.
  \end{align*}
Since $\nuclear{A}\le \sqrt{n}\fnorm{A}$ for any matrix $A$ (\Cref{fact:norm-ineqs}), we conclude that $\nuclear{N_G-\tildeN} \leq \e n$.
\end{proof}
Next, we show that \Cref{thm:nuclear_ub} yields an efficient nuclear norm sparsification algorithm (\Cref{alg:nuclear}), whose analysis yields our first main result, \Cref{thm:nuclear-ub-informal}.

\begin{algorithm2e}[t]
\caption{Additive \nuclearApx}\label{alg:nuclear}
\DontPrintSemicolon
\LinesNumbered
\KwIn{Graph $G = (V,E,w)$ supporting $\mathsf{GetNeighbor}$ queries (\Cref{def:query_model}), accuracy $\e$.}
\KwOut{Nuclear norm sparsifier $\tildeN$ of $G$.}
Initialize $\tildeN = 0$, $c=1$\;
\For{$v \in V$}{
$v' \gets \mathsf{GetNeighbor}(v,c)$ \;
 \While{$v'\neq \emptyset \text{ and } w_{\{v,v'\}}\ge \frac{\e^2}{2} \deg(v)$}{
 \If{$w_{\{v,v'\}}\ge \frac{\e^2}{2} \deg(v')$}{
    Set $\tildeN(v,v') = \tildeN(v',v) = w_{\set{v,v'}}/(\sqrt{\deg(v)\deg(v')})$
    \; 
    Set $\tildeN(v',v) = \tildeN(v,v')$\;  
 }
 $c\gets c+1$, $v' \gets \mathsf{GetNeighbor}(v,c)$ \;
 }
}
\textbf{Return:} $\tildeN$
\end{algorithm2e}

\begingroup
\renewcommand{\proofname}{Proof of \Cref{thm:nuclear-ub-informal}}
\begin{proof} 
We first claim that \Cref{alg:nuclear} returns exactly the matrix $\tilde{N}$ described in \Cref{thm:nuclear_ub}. In particular, it suffices to show that the while loop is executed for all $\{v,v'\}\in E'$, where $E'$ is as in \Cref{thm:nuclear_ub}. To see why this is the case, observe that, because edges are processed in decreasing order of weight, when executing the for loop for vertex $v$, the inner while loop is executed for all $\{v,v'\}$ such that $w_{\{v,v'\}}\ge \frac{\e^2}{2} \deg(v)$. Because $\max\set{\deg(v'),\deg(v)} \geq \deg(v)$, the loop executes for a superset of the edges in $E'$ that have $v$ as an endpoint.

Next, we bound the runtime. There are $n$ iterations of the for loop, each of which requires an $O(1)$ time oracle call. Additionally, every iteration adds $2$ non-zero entries to $\tilde{N}$ at the cost of one oracle call and other $O(1)$ time operations. So, the total runtime is $O(n + \nnz \tildeN)$. As shown in Theorem~\ref{thm:nuclear_ub}, $\nnz \tildeN = O(n \epsilon^{-2})$. Consequently, the algorithm runs in $O(n\epsilon^{-2}).$
\end{proof}
\endgroup

Note that the nuclear sparsifier $\tildeN$ returned by \Cref{alg:nuclear} is \textit{not} guaranteed to be the normalized adjacency matrix of any undirected graph. As access to a graphical sparsifier may be desirable, we show how to modify \Cref{alg:nuclear} to obtain a graphical nuclear sparsifier in \Cref{lemma:graphical-nuclear-ub}. 

\subsection{Applications to Spectral Density Estimation}
An important application of our nuclear norm approximation method from \Cref{alg:nuclear} is to develop faster algorithms for SDE (Problem~ \ref{prob:spectral_density}). In particular, from \Cref{lem:redx-sde}, if we compute $\tildeN$ such that  $\|{N-\tildeN}\|_*\le \frac{\e}{2} n$, then computing an $\frac{\epsilon}{2}$ accurate SDE for $\tildeN$ immediately yields an $\epsilon$ accurate SDE for $N_G$. If we compute the SDE for $\tildeN$ using the existing linear time method from \citet{BravermanKrishnanMusco:2022}, which runs in roughly $O(\nnz{\tilde{N}}\epsilon^{-1})$ time, then we immediately obtain our \Cref{corr:sde-informal}, i.e., that there is an $O(n\epsilon^{-3})$ time randomized algorithm for approximating the spectral density of any weighted graph $G$. 

\thmsde*

\begin{proof} 
Combining Theorem 1.4 in \citet{BravermanKrishnanMusco:2022} with the discretization procedure described in the appendix of \citet{BravermanKM21Arxiv} yields an algorithm that solves the general matrix version of \Cref{prob:spectral_density} with probability $99/100$ in time $\bigO\left(\nnz{M}\e^{-1}\cdot (1+\e^{-2}n^{-2}\log^3(1/\e)\right)$ for any symmetric matrix $M$. 

Suppose that $n \geq \e^{-2}\log^3(1/\e)$ so that $\e^{-2}n^{-2}\log^3(1/\e) \leq 1$. For such values of $n$, we obtain \Cref{corr:sde-informal} by simply computing an $\frac{\epsilon}{2}$-additive nuclear sparsifier $\tildeN$  for $N_G$ using \Cref{thm:nuclear-ub-informal} in $O(n\epsilon^{-2})$ time, then running the method of \citet{BravermanKrishnanMusco:2022} with error $\frac{\epsilon}{2}$ on $\tildeN$. By \Cref{lem:redx-sde} and triangle inequality, the eigenvalues returned are guaranteed to approximate those of $N_G$ to error at most $\frac{\epsilon}{2} + \frac{\epsilon}{2}$ in Wasserstein distance. Moreover, since $\tildeN$ is guaranteed to have at most $O(n\epsilon^{-2})$ entries, the total runtime in $O(n\epsilon^{-3})$ as required. 

Alternatively, suppose that that $n < \e^{-2}\log^3(1/\e)$. In this case, we can simply use a direct eigendecomposition method to compute the eigenvalues of $N_G$. Doing so takes $\tilde{O}(n^\omega) = \bigO(n^{2.372})$ \citep{BanksGarza-VargasKulkarni:2022}, and we have that $O(n^{2.372}) = O(n\epsilon^{-3})$, which yields the theorem. 
\end{proof}

We can also use \Cref{alg:nuclear}  to obtain the first \emph{deterministic} sublinear time method for spectral density estimation. In particular, we claim:

\thmsdedeter*

\begin{proof}
    The main observation required to prove \Cref{corr:sde-deterministic} is that it suffices to compute the first $O(1/\epsilon)$ (uncentered) moments of the eigenvalues of $\tildeN$ in $n \cdot 2^{O(\frac{1}{\e}\log\frac{1}{\e})}$ time, where $\tildeN$ is the result of \Cref{alg:nuclear} run on $N_G$ with accuracy parameter $\frac{\e}{2}$. In particular, it is well known that, if two distributions $p$ and $q$ have the same first $\frac{2c}{\e}$ (uncentered) moments for a fixed constant $c$, then $W_1(p,q)\leq \epsilon/2$ \cite{KongValiant:2017,BravermanKrishnanMusco:2022}.\footnote{By Lemma 3.1 in \citet{BravermanKrishnanMusco:2022}, $c\leq 36$, although it is likely that this upper bound is loose.} So, given moments of an unknown distribution $p$ (here, the spectral density of $\tilde{N}$), we can find a distribution approximating $p$ in Wasserstein distance by simply returning any distribution that (approximately) matches those moments. Formally, \citet{BravermanKrishnanMusco:2022} show in their Lemma 3.4 that, given $p$'s first $2c/\epsilon$ moments, a distribution approximating $p$ to $\epsilon/2$ error Wasserstein distance can be found in $\poly(1/\epsilon)$ time. Furthermore, by their Theorem B.1, this distribution can be converted to a uniform distribution over a set of $n$ approximate eigenvalues in $O(n + 1/\epsilon)$ time.

    To see that we can efficiently compute the first compute the first $j = \frac{c}{2\e}$ eigenvalue moments of $\tilde{N}$ in $n \cdot 2^{O(\frac{1}{\e}\log\frac{1}{\e})}$ time,  observe that the $j^\text{th}$ moment of $\tildeN$'s eigenvalue distribution equals: 
    \begin{align*}
    \frac{1}{n}\sum_{i=1}^n \tilde{\lambda}_i^j = \Tr{\tildeN^j},
    \end{align*}

    where $\tilde{\lambda}_1, \ldots,\tilde{\lambda}_n$ are $\tildeN$'s eigenvalues. As $\tildeN$ is symmetric with at most $8\epsilon^{-2}$ non-zeros per row/column (guaranteed by \Cref{thm:nuclear_ub}), for any positive integer $k$, $\tildeN^k$ has at most $(8\e^{-2})^k$ non-zeros per row/column. Thus, given $\tildeN^k$, the matrix $\tildeN^{k+1} = \tildeN^k\tildeN$ can be computed in $n(8\e^{-2})^{k+1} = n\cdot 2^{O(k\log\frac{1}{\epsilon})}$ time. So, we conclude that $\tildeN^2, \tildeN^3, \ldots, \tildeN^j$ can be computed, and traces evaluated exactly, in $n\cdot 2^{O(j\log\frac{1}{\epsilon})}$ time. Note that here and in the remainder of the paper, we work in the Real RAM model of computation for simplicity.
\end{proof}

\section{Lower Bounds for Nuclear Sparsification}\label{sec:lower}
In this section, we prove \Cref{lem:norm_adj-informal} which shows that our sparsification result of \Cref{thm:nuclear-ub-informal} is nearly optimal in terms of sparsity.
The proof proceeds in two steps. First, in Section~\ref{subsec:counting-argument}, we restrict to the case where $\epsilon = \Theta(1/\sqrt{n})$, showing that with constant probability, the \emph{unnormalized} adjacency matrix $A$ of a random \Erdos-\Renyi graph cannot be approximated well in the nuclear norm by any sparse matrix $B'$ with $\nnz{B'} = o(n \epsilon^{-2}/\log^2 \epsilon^{-1}) = o(n^2/\log^2 n)$, i.e., $\nuclear{A-B'} = \Omega(n^{1.5})$.

Second, in Section~\ref{subsec:tiling-argument} we extend the argument in Section~\ref{subsec:counting-argument} to a lower bound on nuclear sparsification of \emph{normalized} adjacency matrices for a broader range of $\epsilon$ to obtain Theorem~\ref{lem:norm_adj-informal}. 
Concretely, we show that there exists a fixed constant $\epsilon_0$ such that, whenever $\Theta(1/\sqrt{n}) \leq \epsilon < \epsilon_0$, we can \emph{tile} several random \Erdos-\Renyi graphs in order to construct a graph $G$ that cannot be approximated in nuclear norm by any matrix with $o(n \epsilon^{-2}/\log^2 \epsilon^{-1})$ entries.

\subsection{Sparsification Lower Bound for \Erdos-\Renyi Adjacency Matrices}\label{subsec:counting-argument}
Our sparsity lower bound is based on showing that, with constant probability, all pairs of unnormalized adjacency matrices in a sample of roughly $2^{O(n^2)}$ independent random \Erdos-\Renyi graphs are \emph{far} in nuclear norm -- roughly $n^{1.5}$ far. On the other hand, the set of $s$-sparse matrices can be well-approximated in nuclear norm by a finite set of size approximately $2^s$. So, if $s = o(n^2)$, there are simply not enough choices of $s$-sparse matrices to approximate all graphs in our sample.

In the proof of this result, we leverage the following (simplified) matrix concentration result to show that two \Erdos-\Renyi random graphs tend to be far apart in the nuclear norm.
\begin{fact}[\citet{GuionnetZ00}, Theorem 1.1 with $f(x) = |x|$]\label{thm:concentration} 
  Let $C \in \R^{n \times n}$ be a fixed symmetric matrix whose entries have magnitude at most $1$. Let $X \in \R^{n \times n}$ be a random symmetric matrix with $X_{ij} = C_{ij} \omega_{ij}$, where $\omega_{ij} = \omega_{ji}$, and $\set{\omega_{ij}: 1 \leq i\leq j\leq n }$ are independent random variables supported on $[-1,1]$. Then, for any $\delta > 16 \sqrt{\pi}/n$,
  \[ \prob{ \abs{ \|X\|_* - \E[\|X\|_* ] } \geq n^{1.5} \delta } \leq 4 \cdot \exp\paren{ -\frac{n^2 (\delta-16 \sqrt{\pi}/n)^2}{64} } \mper \]
\end{fact}

We can use \Cref{thm:concentration} to bound the expected value of the nuclear norm distance between the adjacency matrices of two \Erdos-\Renyi $\cG(n, 1/2)$ random graphs:
\begin{lemma}\label{lemma:expectation}
  For $n \geq 10$ let $X$ be the difference between the adjacency matrices of two \Erdos-\Renyi, $\cG(n, 1/2)$ random graphs, i.e., $X \in \set{-1,0,1}^{n \times n}$ has zeros on its diagonal, $X_{ji} = X_{ij}$, and
  \begin{equation}
    X_{ij} \overset{\mathrm{iid}}{\sim}
    \begin{cases}
      -1, & \text{ with probability } 1/4 \\
      0,  & \text{ with probability } 1/2 \\
      1, & \text{ with probability } 1/4
    \end{cases} \mcom \quad \text{ for } i>j \mper
  \end{equation}
  Then, $n^{1.5}/100 \leq \E \nuclear{X} \leq n^{1.5}$.
\end{lemma}

\begin{proof}
From Corollary 3.9 of \citet{BandeiraH16}, we know that $\prob{\spnorm{X} > 5 \sqrt{n}} \leq \exp(-n/4)$. Additionally, $\prob{\fnorm{X}^2 \leq n^2/9} \leq \exp(-n^2/32)$ by Chernoff bound. Conditioning on these two events, we see that, for $n \geq 10$, with probability $1-\exp(-n/5)$, 
  \begin{align*}
      n^2/9 \leq \norm{X}_F^2 \leq \norm{X}_2 \nuclear{X} \leq 5
      \sqrt{n} \nuclear{X}.
  \end{align*}
  I.e., conditioned on the events, $\nuclear{X} \geq n^{1.5}/45$. It follows that $\E\nuclear{X} \geq (1 - \exp(-n/5))\cdot n^{1.5}/45 \geq n^{1.5}/100$. Moreover, we have that $\nuclear{X} \leq \sqrt{n}\fnorm{X} \leq n^{1.5}$, where the last step follows because $X$ is a $\set{-1,0,1}$ matrix.
\end{proof}

As discussed, our lower bound proof also requires a covering number bound for the set of sparse matrices with bounded entries under the nuclear norm metric. We provide such a bound below: 
\begin{lemma}\label{lem:discretization_sparse}
    For any constant $c < 1$, let $\mathcal{B}_{c} \subset \R^{n\times n}$ denote the set of matrices with at most $cn^2/\log^2 n$ non-zero entries, each of which is bounded in magnitude by $n^2$. Then there is a finite set of matrices $\cS_c \subset \mathcal{B}_{c}$ of size $|\cS_c| \leq \exp\paren{ \frac{cn^2}{\log n}\paren{4 + \frac{\log(1/c)}{\log n}} }$ such that, for all $B\in \mathcal{B}_{c}$, there is an $S \in \cS_c$ such that $\nuclear{B-S} \leq 1/\sqrt{n}$.
\end{lemma}
\begin{proof}
    Let $\cS_c$ contain any matrix $S \in \R^{n \times n}$ such that $\nnz{S} \leq cn^2/\log^2 n$ and $\forall i,j \in [n]$, $S_{ij} \in \set{-n^2, -(n^2-\e),-(n^2-2\e), \dots,(n^2-2\e), (n^2-\e),n^2}$ for $\e = 2/n$. We have that:
  \begin{align*}
     \abs{\mathcal{S}_c} & \leq \binom{n^2}{\frac{cn^2}{ \log^2 n}}  \cdot \paren{\frac{2n}{\e}}^{\paren{\frac{cn^2}{  \log^2 n}}} 
     \leq {( e/c \cdot \log^2n )}^{\paren{\frac{cn^2}{  \log^2 n}}}  \cdot n^{\paren{\frac{2cn^2}{  \log^2 n}}} \\
     & = \exp\paren{ \frac{cn^2 \log (1/c)}{ \log^2 n} +  \frac{cn^2}{ \log^2 n} + \frac{2cn^2 \log\log n}{ \log^2 n} + \frac{2cn^2 \log n}{ \log^2 n} } 
     \leq \exp\paren{ \frac{cn^2}{\log n}\paren{4 + \frac{\log(1/c)}{\log n}} } \mper
  \end{align*}
  Moreover, since $\e = 2/n$, for any $B\in \mathcal{B}_c$, we have a matrix $S \in \mathcal{S}_{c}$ with $\norm{S-B}_F^2 \leq n \cdot (\epsilon/2)^2 \leq 1/n$. It follows that $\nuclear{S-B} \leq \sqrt{n}\norm{S-B}_F \leq 1/\sqrt{n}$.
\end{proof}

With the above lemmas in place, we are now ready to prove \Cref{prop:count_lb}, which shows that, with good probability, the unnormalized adjacency matrix of an \Erdos-\Renyi graph cannot be approximated well in the nuclear norm by sparse matrices.
\begin{theorem}\label{prop:count_lb}
  Let $A$ be the (unormalized) adjacency matrix of an \Erdos-\Renyi random graph $\cG(n,1/2)$, where $n \geq C$ for a universal constant $C$. For a universal constant $c$, with probability at least $3/4$, there is no matrix $B$ with $\nnz{B} \leq cn^2/\log^2 n$ such that
  $\nuclear{A-B} \leq n^{1.5}/500$.
\end{theorem}
\begin{proof}
  Let $A$ and $A'$ be drawn from $\cG(n,1/2)$. Then $A-A'$ follows the distribution of the matrix $X$ in \Cref{lemma:expectation}. Therefore, $\E \nuclear{A-A'} \geq n^{1.5}/100$. \Cref{thm:concentration} then implies that 
  \[ \prob{ \abs{ \nuclear{A-A'} - \E \nuclear{A-A'} }\geq \delta n^{1.5} } \leq 4 \exp\paren{ \frac{-n^2 (\delta-30/n)^2}{64} } \mper \]

 Choosing $\delta = 1/200$ in the above bound, yields that for $n \geq 3000$, 
  \[ \prob{  \nuclear{A_1-A_2} \leq \frac{n^{1.5}}{200} } \leq 4 \exp\paren{\frac{-n^2}{4 \cdot 10^6} }. \]

  Let $\cA \defeq \{A_1,\dots,A_m\}$ be a set of adjacency matrices of $m$ independent $\cG(n,1/2)$ graphs. 
  By union bound,
  $ \prob{ \exists i \neq j: \nuclear{A_i-A_j} \leq \frac{n^{1.5}}{200}} \leq 4 {\binom{m}{2}} \cdot \exp\paren{\frac{-n^2}{4 \cdot 10^6} }$.

 Setting $m = \exp\paren{n^2/(16 \cdot 10^6)}$, the above equation then yields that
  \begin{align}
  \label{eq:cond_event}
      \prob{\exists i \neq j: \nuclear{A_i-A_j} \leq \frac{n^{1.5}}{200}} \leq 2 \exp\paren{\frac{-n^2}{8 \cdot 10^6} } \mper
  \end{align}

  The above equation implies that, with high probability, $m = e^{\Omega(n^2)}$ random $\cG(n,1/2)$ graphs will have adjacency matrices, $A_1, \ldots, A_m$, that are all more than $n^{1.5}/200$ apart in nuclear norm from each other. Condition on this event in the remainder of the proof.

  Our goal is to show that a large fraction of $A_1, \ldots, A_m$ cannot be well-approximated by a matrix with $\leq cn^2/\log^2 n$ non-zero entries.

 Specifically, suppose for $z$ matrices $A_{j_1}, \ldots, A_{j_z}$, there are matrices $B_1, \ldots, B_z$, each with at most $cn^2/\log^2 n$ non-zero entries, such that $\|B_i - A_{j_i}\|\leq n^{1.5}/500$. Observe that, since $\|A_{j_i} - A_{j_k}\|_* > n^{1.5}/200$ for all $i,k\in [z]$, by triangle inequality, we have that $\|B_i - B_k\|_* > n^{1.5}/1000$ for all $i,k\in [z]$. Moreover, we know that every entry in each $B_i$ is bounded in magnitude by $n^{1.5}$. If it was not, then since every entry in $\cA_{j_i}$ is bounded by $1$, we would have $\nuclear{A_{j_i} -B_i} \geq \norm{A_{j_i} -B_i}_F \geq n^{1.5} -1$.
 
 Then, applying \Cref{lem:discretization_sparse}, for every $B_i$, $i\in [z]$, there is a matrix $S_i \in \mathcal{S}_c$ such that $\|S_i - B_i\|_* \leq 1/\sqrt{n} \leq n^{1.5}/2000$. By triangle inequality, since $\|B_i - B_k\|_* > n^{1.5}/1000$ for all $i,k\in [z]$, it must be that $S_i \neq S_k$ for all $i,k\in [z]$. Accordingly, we have that $z \leq |\mathcal{S}_c|$. So at most $|\mathcal{S}_c|$ of the $m$ matrices from $A_1, \ldots, A_m$ can be approximated by a matrix with $cn^2/\log^2 n$ non-zeroes, where $|\mathcal{S}_c|$ is as defined in  \Cref{lem:discretization_sparse}.

Setting $c = 1/(16 \cdot 10^6)$, we have that $m = \exp\paren{cn^2}$ and $\abs{\mathcal{S}_c} \leq \exp\paren{ \frac{7c n^2}{ \log n} }$. Combined with \eqref{eq:cond_event}, it follows that, if we select $m$ adjacency matrices at random from  $\cG(n,1/2)$, the expected fraction that can be approximated to nuclear norm error $n^{1.5}/500$ with a $cn^2/\log^2 n$ sparse matrix is upper bounded by:
\begin{align*}
    2 \exp\paren{\frac{-n^2}{8 \cdot 10^6}} + \left(1 - 2 \exp\paren{\frac{-n^2}{8 \cdot 10^6}}\right)\cdot \frac{\exp\paren{ \frac{7c n^2}{ \log n} }}{\exp\paren{cn^2}}. 
\end{align*}
For sufficiently large $n$, this is upper bounded by $1/4$.
\end{proof}

\subsection{Proof of Theorem~\ref{lem:norm_adj-informal}}\label{subsec:tiling-argument}

\Cref{subsec:counting-argument} proves the \Cref{lem:norm_adj-informal} for the special case of $\e = \Theta(1/\sqrt{n})$, but for unnormalized adjacency matrices. In this section, we extend the \Cref{prop:count_lb} for larger values of $\e$ and for normalized adjacency matrices to prove the \Cref{lem:norm_adj-informal}. Our lower bound construction follows by considering a block-diagonal matrix of size $n \times n$, with $\floor{n/b} \geq 1$ many $b$-by-$b$ blocks and each block satisfying the properties in~\Cref{prop:count_lb}.

In the following lemma, we show that the nuclear norm of a block diagonal matrix is at least the sum of the nuclear norms of the blocks of the matrix.

\begin{lemma}[Block-wise Nuclear Norm Bound]
\label{lem:diag_nuclear_bound}
For any symmetric $A \in \R^{n \times n}$ and $S_1,\ldots,S_k \subseteq [n]$ that partition $[n]$, i.e. $\cup_{i \in [k]} S_i = [n]$ and $S_i \cap S_j = \emptyset$ for all $i \neq j$, it is the case that $\nuclear{A} \geq \sum_{i \in {k}} \nuclear{A_{S_i,S_i}}$, where $A_{S_i,S_i}$ denotes the principal submatrix of $A$ indexed by $S_i$.
\end{lemma}

\begin{proof}
By \Cref{fact:norm-ineqs} we know that there are symmetric $X_i^* \in \bbS^{\abs{S_i} \times \abs{S_i}}$ for all $i \in [k]$ such that $\norm{X_i}_2 \leq 1$ and $\inprod{A_{S_i,S_i},X_i^*} = \nuclear{A_{S_i,S_i}}$ for all $i \in [k]$. Let $\hat{X} \in \bbS^{n \times n}$ be the block-diagonal matrix defined as $\hat{X}_{S_i,S_i} = X_i$ for all $i \in [k]$ and all other entries of $Y$ set to $0$. Note that for all $x \in \R^{n}$,
\[
|x^\top \hat{X} x|
= \left|\sum_{i \in [k]} x^\top_{S_i} X_i x_{S_i} \right|
\leq \sum_{i \in [k]} \norm{x_{S_i}}_2^2 = \norm{x}_2^2\,
\]
where we used that $|x^\top_{S_i} X_i x_{S_i}| \leq \norm{x_{S_i}}_2^2$ since $\norm{X_i}_2 \leq 1$. Consequently, $\norm{\hat{X}}_2 \leq 1$ and the result follows from \Cref{fact:norm-ineqs} as 
\[
 \nuclear{A} = \max_{ X \in \bbS^{n \times n}:\spnorm{X}\le 1}\inprod{A,X}
 \leq \inprod{A,\hat{X}}
 = \sum_{i \in [k]} \inprod{A_{S_i,S_i},X_i^*} = \sum_{i \in [k]} \nuclear{A_{S_i,S_i}}\,. \qedhere
\]
\end{proof}

Next, we show that if the blocks of a block diagonal matrix cannot be well approximated by a sparse matrix, then the block diagonal matrix also cannot be well approximated by a sparse matrix.

\begin{corollary}[Corollary of \Cref{lem:diag_nuclear_bound}]\label{prop:tiling_sparsity_count}
  Let $n,k,b \in \N$ be such that $k = \floor{n/b} \geq 1$. Let $E \in \R^{n \times n}$ be a block diagonal matrix $E = \mathrm{diag}(E_1,E_2,\dots, E_k,R)$, where $E_1,\dots,E_k \in \R^{b \times b}$ and $R$ is an arbitrary matrix of size $(n-kb)$. For any pair of constants $\tau,\eta>0$, let $E_i$, for each $i \in [k]$, be such that for any sparse matrix $B'$ with $\nnz{B'} \leq \eta$, $\nuclear{E_i - B'} > \tau$. 
  Then, for any matrix $B$ with $\nnz{B} \leq (k/2) \cdot \eta $, we have that $\nuclear{E-B}   > n \tau/(4b)$.
\end{corollary}

\begin{proof}
  Let $B$ be any $n \times n$ matrix with $\nnz{B} \leq (k/2) \cdot \eta$. For $i \in [k]$, define $B_i$ as the submatrix of $B$ with entries corresponding to $E_i$. 
  Let $\calI \defeq \{i\in[k]:\nnz{B_i}\le \eta\}$ be the index set of \emph{sparse} blocks $B_i$. We have that $|\calI|\ge \lfloor k/2\rfloor$. This is because if $|\calI| < \lfloor k/2\rfloor$, then  $\nnz{B} \geq  \sum_{i\in[k]\setminus\calI}\nnz{B_i}> (k/2) \cdot \eta$, leading to a contradiction. Therefore, for each $i \in \cI$, 
  $\nuclear{E_i-B_i} \geq \tau$. Therefore, by \Cref{lem:diag_nuclear_bound}, we have that $\nuclear{E-B} \geq (k/2) \tau \geq  n \tau /(4b) $.
\end{proof}

Using \Cref{prop:tiling_sparsity_count}, we get the following lower bound on the sparsifiability of normalized adjacency matrices in the nuclear norm by setting $b = \bigo{1/\e^2}$ and noting that (from \Cref{prop:count_lb}) \Erdos-\Renyi random graphs satisfy the properties of the matrices $E_1,\dots E_k$ in \Cref{prop:tiling_sparsity_count}. 

\sparsitynuclear*

\begin{proof}
  Let $G_1,\dots,G_k$ denote the \Erdos \Renyi random graphs of size $b \geq C$ from \Cref{prop:count_lb}. Let $A_1,\dots,A_k$ denote the adjacency matrix of $G_1,\dots,G_k$ respectively. 
  We define the graph $G$ to be a graph of size $n$ which is a disjoint union of $G_1,\dots,G_k$, where $k = \floor{n/b} \geq 1$, and an arbitrary graph with adjacency matrix $R$ on $(n-kb)$ vertices, such that the degree of each of the $(n-kb)$ nodes is at least $1$.
  Let $A_G$ denote the adjacency matrix of $G$. We get that $A_G = \mathrm{diag}(A_1,A_2,\dots, A_k,R)$.
  Note that the matrix $A_G$ satisfies the conditions of the matrix $E$ from \Cref{prop:tiling_sparsity_count}, with $E_i=A_i$, $i \in [k]$, $\eta = c b^2/\log^2 b$, and $\tau =  b^{1.5}/500$. This is true because \Cref{prop:count_lb} guarantees that, for each $i \in [k]$, any matrix $B'$ with $\nnz{B'} \leq c b^2/\log^2 b$, $\nuclear{A_i-B'} > c'' b^{1.5}$, where $c'' = 1/500$.
   
  From \Cref{prop:tiling_sparsity_count} and the dual norm characterization of the nuclear norm (\Cref{fact:norm-ineqs}), we get that for any matrix $B$ with $\nnz{B} \leq \frac{n}{2b} (c b^2/\log^2 b)$, there exists a matrix $\tX$ with $\spnorm{\tX} \leq 1$ such that 
    \[ \nuclear{A_G-B} = \inprod{A_G-B, \tX} > \frac{c'' n \sqrt{b} }{4} \mper \]
  Since the graphs $G_1,\dots,G_k$ are \Erdos-\Renyi random graphs of size $b$, the degree matrix $D_G$ of graph $G$ satisfies $D_G  \preccurlyeq \mathrm{diag}( {b}\cdot \mathbf{I}, \cdots, b \cdot \mathbf{I},1,\cdots, 1)$ where each $\mathbf{I}$ is a $b$-by-$b$ identity matrix, we have that 
  \begin{align*}
    \nuclear{D_G^{-1/2}(A_G-B)D_G^{-1/2}}\geq \inprod{D_G^{-1/2}(A_G-B)D_G^{-1/2}, \tX  }  = \inprod{(A_G-B), D_G^{-1/2} \tX D_G^{-1/2} } 
    > \frac{c'' n\sqrt{b}}{4b} \mcom
  \end{align*}
  where in the last inequality we used the definition of the inner product and the fact that all the degrees are less than $b$.
    
  Set  $b = (c'')^2/(4\e^2)$. For a sufficiently small constant $\e_0$, $b \geq C$ for all $\e \leq \e_0$. We also need to ensure that $n/b \geq 1$, which implies that $n \geq (c'')^2/(4\e^2)$. From the equation above and the value of $b$, we get that for all $\e \leq \e_0$, 
  \[ \nuclear{D_G^{-1/2}(A_G-B)D_G^{-1/2}} = \nuclear{N_G-D_G^{-1/2}BD_G^{-1/2}}   > n\e \mper \] 
  Note that $\nnz{B} = \nnz{D_G^{-1/2}B D_G^{-1/2}}$, since $D_G$ is a diagonal matrix. 
  Since $B$ was an arbitrary matrix with just sparsity constraint, we can represent any arbitrary matrix $M$ with the same sparsity as $B$, by $M = D_G^{-1/2}B D_G^{-1/2}$. Hence, we get that for any matrix $M$ with
  \[ \nnz{M} \leq \paren{\frac{c' (c'')^2}{32 \log^2\paren{\frac{c''/2}{\e}}}} \cdot \frac{n}{\e^2} , \qquad \nuclear{N_G - M} > n\e \mper \qedhere  \]
\end{proof}

\section{Nuclear Sparsification Query Complexity Lower Bound}\label{sec:lower-query}

In this section, we prove \Cref{lem:norm_query-informal}, which asserts that our $O(n/\epsilon^2)$ query algorithm from \Cref{sec:upper} also achieves optimal query complexity for nuclear approximation, up to polylogarithmic factors in $n$. 
Interestingly, our lower bound applies in an {even more general } query model which supports edge queries in addition to neighbor queries. We formalize this model below:

\begin{definition}[Generalized adjacency query model]\label{def:query_model_gen} We say we have \emph{generalized adjacency query access} to a weighted graph $G=(V,E,w)$ if, in addition to the access under the adjacency query model (Definition~\ref{def:query_model}), there is an $O(1)$ time procedure, $\mathsf{GetEdge}(u, v)$ that when queried with any $u,v \in V$ returns $\mathsf{True}$ if $\{u, v\} \in E$ and $\mathsf{False}$ otherwise.
\end{definition}

Our main result in this section is that any deterministic algorithm that makes $o(n\epsilon^2/\log^2\e^{-1})$ queries to the generalized adjacency query model (Definition~\ref{def:query_model_gen}) is unable to distinguish between a random \Erdos-\Renyi graph and the graph which is the complement restricted to the edges which were not queried. \Cref{lem:norm_adj-informal} then follows by an application of Yao's principle. 

In the remainder of the section, we let $Q_{\Alg}(G)$ denote the sequence of queries a deterministic algorithm $\Alg$ makes to the generalized adjacency query model on input $G$. We let $E_{\Alg}(G)$ to denote the set of edges revealed by these queries. That is, for each query of the form $\mathsf{GetEdge}(u, v) \in Q_{\Alg}(G)$, we include $\{u, v\} \in E_{\Alg}(G)$; meanwhile, for each query of the form  $\mathsf{GetNeighbor}(a,i)\in Q_{\Alg}(G)$, we include the edge $\{a,b\}$ returned by the generalized adjacency query model (if any) in $E_{\Alg}(G)$.
Our formal result is as follows:
\begin{lemma}\label{lemma:query_lb_distribution} There exists an $\e_0 \in (0, 1)$ and an absolute constant $C'$ such that for any $\e \in (0, \e_0)$, there exists a distribution $q$ over graphs on $n \geq C'\epsilon^{-2}$ nodes such that the following holds: let $\Alg$ be any deterministic algorithm that on any $n$-node unweighted graph $G = (V, E)$ makes at most $|Q_{\Alg}(G)| \leq  Cn\e^{-2}/\log^2 \e^{-1}$ queries to the generalized adjacency query model (Definition~\ref{def:query_model_gen}) for an absolute constant $C$ and outputs some $X_{\Alg}(G) \in \R^{n \times n}$. Then
\begin{align*}
    \ProbOp_{G\sim q}\{\norm{X_{\Alg}(G) - N_G}_{*} > n\e\} > 1/4.
\end{align*}
\end{lemma}
\begin{proof}
We describe a distribution $q$ over graphs on $n = 2kb \geq n_0$ vertices, where $b = \floor{C' \epsilon^{-2}}$ for a sufficiently small constant $C'$ and $k = \ceil{n_0/2b}$. 

\paragraph{The hard distribution:}  
Specifically, the vertex set of our graph will be split into $2k$ sets of size $b$, which we denote by $V^{1, 1}, V^{1, 2}, V^{2, 1}, V^{2, 2}, \ldots, V^{k, 1}, V^{k, 2}$. 
For $r \in [k]$, let $v^{r,1}_{j}$ and $v^{r,2}_{j}$ denote the $j$-th vertex in $V^{r, 1}$ and $V^{r, 2}$, respectively. 
For each $i \neq j \in [b]$ and $r \in [k]$, we draw $X_{i,j, r} \sim \text{Ber}(1/2)$. If $X_{i,j,r} = 1$, we add the edges $\{v_i^{r,1}, v_j^{r,1}\}, \{v_i^{r,2}, v_j^{r,2}\}$ to $E$. If $X_{i,j, r} = 0$, we add the edges $\{v_i^{r,1}, v_j^{r,2}\}, \{v_i^{r,2}, v_j^{r,1}\}$ to $E$. 
Note that any graph drawn from this distribution is degree $d \defeq 2(b-1)$-regular. 

\paragraph{Indistinguishable graph pairs:}  
For any realization $G = (V,E) \sim q$, let $G' = (V, E')$ denote the complement of $G$. 
That is, for any nodes $u,v$, $(u,v)\in E'$ if and only if $(u,v)\notin E$. Observe that $G'$ is itself distributed as $q$. 
Additionally, define $\bar{G}= (V, \Ebar)$ as follows: Initialize $\Ebar = E'$. Then, for all $i,j,r$, if either $\{v_i^{r,1}, v_j^{r,2}\}$ or $\{v_i^{r,2}, v_j^{r,1}\}$ are in $E_{\Alg}(G)$, then delete both $\{v_i^{r,1}, v_j^{r,1}\}$ $\{v_i^{r,2}, v_j^{r,2}\}$ from $E'$ and add $\{v_i^{r,1}, v_j^{r,2}\}$ and $\{v_i^{r,2}, v_j^{r,1}\}$; likewise if either $\{v_i^{r,1}, v_j^{r,1}\}$ or $\{v_i^{r,2}, v_j^{r,2}\}$ are in $E_{\Alg}(G)$, then delete both $\{v_i^{r,1}, v_j^{r,2}\}$ and $\{v_i^{r,2}, v_j^{r,1}\}$ from $E'$ and add $\{v_i^{r,1}, v_j^{r,1}\}$ and $\{v_i^{r,2}, v_j^{r,2}\}$.
In other words, $\bar{G}$ is the complement of $G$ on all edges \emph{except} those revealed by the queries in $Q_{\Alg}(G)$. Our proof of \Cref{lemma:query_lb_distribution} hinges on two important observations:
\begin{enumerate}
    \item If $G\sim q$, then $\bar{G}$ is also distributed according to $q$.
    \item For any $G$, $Q_{\Alg}(G) = Q_{\Alg}(\bar{G})$, and consequently $X_{\Alg}(G) = X_{\Alg}(\bar{G})$. 
\end{enumerate}
The last claim follows from the fact that, not only are $G$ and $\bar{G}$ identical on any edge accessed by a $\mathsf{GetEdge}$ or $\mathsf{GetNeighbor}$ query, but since both graphs are regular with the same degree, the degree information returned by any $\mathsf{GetNeighbor}$ also matches.

With the above claims in place, it will suffice to show that, for any deterministic algorithm that makes a small number of queries, $\bar{G}$ and ${G}$ are \emph{far} in the nuclear norm. Accordingly, $\Alg$ cannot well approximate both graphs simultaneously. Concretely, we will show that there exists a set $\cH \subset \supp(q)$ such that $\probSub{G \in \cH}{G \sim q} \geq 17/32$ and, for every $G \in \cH$, 
\begin{align}\label{eq:indistinguishable}
    \max \left\{\norm{N_{\bar{G}} - X_{\Alg}(\bar{G})}_*, \norm{N_{G} - X_{\Alg}(G)}_*\right\} \geq \frac{1}{2} \norm{N_{\bar{G}} - N_G}_* \geq n\e.
\end{align}

In the remainder of the proof, we let $G[S]$ denote the vertex-induced subgraph of $G$ corresponding to $S\subset V$. 
Observe that $A_{G[V^{r, 1}]} \sim \mathcal{G}(b, 1/2)$ for each $r \in [k]$. For any pair of constants $\tau, \eta$, let $\cT_{\tau, \eta}$ denote the set of all graphs $H$ on $n_H$ nodes such that 
\begin{align*}
    \norm{A_{H} - B'}_* &\geq \tau {n_H}^{1.5} \text{ for all $B'$ with $\nnz {B'} \leq \eta {n_H}^2/\log^2(n_H)$}.
\end{align*}
\Cref{prop:count_lb} ensures that if $R$ is an \Erdos-\Renyi $\mathcal{G}(n, 1/2)$ graph, then $R \in \cT_{c, c'}$ for fixed constants $c, c'$ with probability at least $3/4$. 

Let $\mathcal{H} \defeq \left\{G \in \supp(q) : A_G[V^{r, 1}] \in \cT_{c, c'} \text{ for more than $9k/16$ values of } r \in [k]\right\}$. By the independence of the $k$ components in $G$ and a Chernoff bound, we have that, for sufficiently large $k$ (guaranteed by having $n = \Omega(1/\epsilon^2)$), 
$\probSub{G \in \mathcal{H}}{G \sim q} \geq 17/32$. 

Now, suppose $G \in \mathcal{H}$ and, without loss of generality, assume that $G[V^{r, 1}] \in \cT_{c, c'}$ for all $r \in [\ceil{9/16 \cdot k}]$. The $r$-th block of $N_G - N_{\bar{G}}$ has the following form for $r \in [\ceil{9/16 \cdot k}]$: 
\begin{align}\label{eq:backref}
    N_{G[V^{r, 1} \cup V^{r, 2}]} - N_{\bar{G}[V^{r, 1} \cup V^{r, 2}]}= \frac{1}{2(b-1)} \paren{\paren{A_{G[V^{r, 1} \cup V^{r, 2}]} - A_{G'[V^{r, 1} \cup V^{r, 2}]}} - {P_r}}, 
\end{align}
where $P_r \defeq (A_{\bar{G}[V^{r, 1} \cup V^{r, 2}]} - A_{G'[V^{r, 1} \cup V^{r, 2}]})$. $P_r$ is the edge perturbation matrix in the $r$-th block corresponding to the edges revealed by the queries in $Q_{\Alg}(G)$. 

Applying \Cref{lem:diag_nuclear_bound} guarantees that for $r \in [\ceil{9/16 \cdot k}]$, $G[V^{r, 1} \cup V^{r, 2}] \in \cT_{c, c'/4}$ (using the fact that $\abs{V^{r, 1} \cup V^{r, 2}}^2 = 4\abs{V^{r, 1}}^2$.) We claim that, for all $B'$ with  this further implies that, for all $B'$ with $\nnz {B'} \leq (c'/4) (2b)^2/\log^2(2b)$,
\begin{align}
\label{eq:er_difference}
    \left\|\left(A_{G[V^{r, 1} \cup V^{r, 2}]} - A_{G'[V^{r, 1} \cup V^{r, 2}]}\right)\right\|_* \geq cb^{1.5}. 
\end{align}
To see that this is the case, observe that $A_{G[V^{r, 1} \cup V^{r, 2}]} - A_{G'[V^{r, 1} \cup V^{r, 2}]} = 2A_{G[V^{r, 1} \cup V^{r, 2}]} - Z$, where $Z = 1_{2b\times 2b} - I$, where $1_{2b\times 2b}$  denotes a $2b\times 2b$ all-ones matrix and $I$ an identity matrix. 
So by triangle inequality,
\begin{align*}
    \norm{(2A_G - Z) - B'}_* \geq \norm{2A_G - B'}_* - \norm{Z}_* \geq 2cb^{1.5} - 2(b-1),
\end{align*}
which is greater than $cb^{1.5}$ for sufficiently large $b$.

Now, let $P$ be the block-diagonal matrix containing $P_r$ in the $r$-th block. Now, we claim that because $Q_{\Alg}(G)$ contains at most $C n\epsilon^{-2} / \log^2(\epsilon^{-1})$ queries for sufficiently small $C$, 
\begin{align*}
    \nnz P \leq 4 \cdot \abs{Q_{\Alg}(G)} \leq C n\epsilon^{-2} / \log^2(\epsilon^{-1}) \leq c'' kb^2/\log^2(b). 
\end{align*}
for any fixed constant $c''$. This is because each query in $Q_{\Alg}(G)$ corresponds to at most one edge in $E_{\Alg}(G)$ and therefore corresponds to at most two edge changes (one edge deleted from and one edge added) to $\bar{G}$ relative to $G'$ and correspondingly, corresponds to four entries modified in the adjacency matrix of $\bar{G}$ relative to $G'.$ Consequently, by \Cref{prop:tiling_sparsity_count} and \eqref{eq:er_difference}, we have that:
\begin{align*}
    \norm{A_{G\left[\bigcup_{r \in [\ceil{9k/16}]} V^{r, 1} \cup V^{r, 2}\right]} - A_{G'\left[\bigcup_{r \in [\ceil{9k/16}]} V^{r, 1} \cup V^{r, 2}\right]} - P}_* \geq \frac{9}{16} k \cdot b^{1.5} \geq C' n \sqrt{b}.
\end{align*}
Finally, by \eqref{eq:backref} and \Cref{lem:diag_nuclear_bound}, we that for sufficiently small $C'$,
\begin{align*}
    \norm{N_{G} - N_{\bar{G}}}_* \geq \frac{C' n \sqrt{b}}{2(b-1)} > 2n\e.
\end{align*}

Therefore, for any $G \in \cH$, either $\norm{N_{G} - X_{\Alg}(G)}_* > n\e$ or $\norm{N_{\bar{G}} - X_{\Alg}(\bar{G})}_* > n\e.$
Since $q$ is uniform over its support, it assigns the same probability to both $\bar{G}$ and $G$. Hence, we conclude that 
\begin{align}
    \prob{\norm{N_{G} - X_{\Alg}(G)}_* > n\epsilon \mid G \in \cH} \geq 1/2.
\end{align}
The lemma now follows by the law of total probability:
\begin{align*}
    \prob{ \norm{N_{G} - X_{\Alg}(G)}_* > n\epsilon} \geq \prob{ \norm{N_{G} - X_{\Alg}(G)}_* > n\epsilon \mid G \in \cH} \prob{G \in \cH} \geq \frac{1}{2} \cdot \frac{17}{32}. &\qedhere
\end{align*}
\end{proof}

With \Cref{lemma:query_lb_distribution} in place, our main query complexity lower bound for nuclear sparsification follows directly from Yao's min-max principle. 

\querynuclearlb*
\begin{proof} 
Any randomized algorithm that makes at most $Q \leq Cn\epsilon^{-2}\log^{-2}(\epsilon^{-1})$ for a fixed constant $C$ can be formulated as a distribution over deterministic algorithms that each makes at most $Q$ queries. As such, if a randomized algorithm succeeds with probability $> 3/4$ on inputs from the hard distribution $q$ guaranteed by Lemma~\ref{lemma:query_lb_distribution}, there must be a deterministic algorithm that succeeds with probability $> 3/4$ and uses an equal number of queries. Such an algorithm is ruled out by Lemma~\ref{lemma:query_lb_distribution}, so the theorem follows by contradiction.
\end{proof}
\section{Separation from Additive Spectral Sparsifiers}\label{sec:additive-spectral}

In this section, we discuss query lower bounds for deterministic algorithms that produces an $\epsilon$-additive spectral sparsifiers (Definition~\ref{def:eps-additive-spectral-spars}). Concretely, here we prove an $\Omega(n^2)$ lower bound on the number of queries to the generalized adjacency query model required by any deterministic $\epsilon$ \spectralApx algorithm on undirected graphs. In this section, for a graph $G = (V, E)$, we use $\bar{L}_G \defeq I - N_G$ to denote its normalized Laplacian.

\begin{figure}[htbp]
    \centering

    \begin{subfigure}[b]{\textwidth}
        \centering 
        \includegraphics[scale=.35]{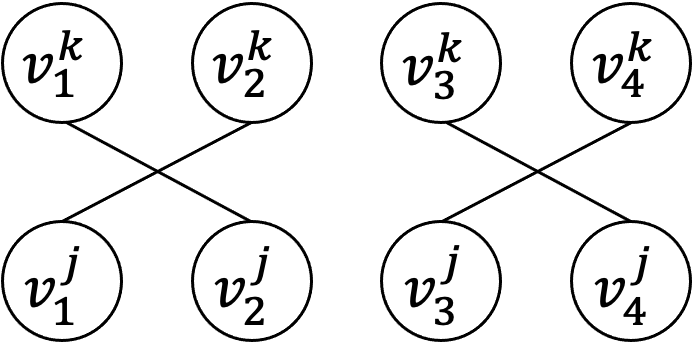}
        \caption{Any edge included in $\tilde{E}$ will always be one of the four edges depicted above for some $k, j \in [m]$. For each $k, j \in [m]$, $G_1$ will contain the edges depicted here between the vertices $\{v_i^k, v_t^j\}_{i,t \in [4]}$ -- and these will be the only edges incident to these 8 vertices.  }
        \label{fig:SL_full1}
    \end{subfigure}
    \begin{subfigure}[b]{\textwidth}
        \centering 
        \includegraphics[scale=.35]{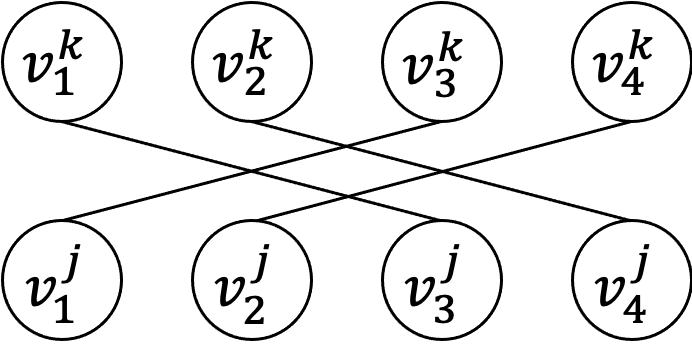}
        \caption{ If for some $k, j \in [m]$, $\tilde{E}$ does not contain the edges depicted in Figure~\ref{fig:SL_full1}, then $G_2$ will contain the edges depicted here between the vertices $\{v_i^k, v_t^j\}_{i,t \in [4]}$ -- and these will be the only edges incident to these 8 vertices.}
        \label{fig:SL_full2}
    \end{subfigure}
    \caption{Illustration of construction of $\tilde{E}, G_1$ and $G_2$ in the proof of Theorem~\ref{lemma:spectral_deterministic_lb}.}
    \label{fig:illustration_lb}
\end{figure}

Our proof proceeds by constructing a \emph{resisting oracle}. We will show that whenever the deterministic algorithm makes $o(n^2)$ queries to the generalized adjacency model, there exist two graphs $G_1, G_2$ on $n$ nodes that have constant nuclear norm distance, along with an oracle (Algorithm~\ref{alg:oracle}) that is \emph{simultaneously} a valid implementation of the generalized adjacency query model (Definition~\ref{def:query_model_gen}) on \emph{both} $G_1$ and $G_2$. Consequently, no deterministic algorithm that makes $o(n^2)$ queries to a generalized adjacency model can distinguish between $G_1$ and $G_2$ -- and therefore must fail to output a nuclear sparsifier on either $G_1$ or on $G_2.$

\begin{algorithm2e}[t]
\caption{Resisting Oracle}\label{alg:oracle}
\DontPrintSemicolon
\LinesNumbered
\KwIn{Parameter $m \in \N$ and a stream of queries $q_1, ..., q_T$}
\tcp{For $\ell \in [T]$, $A^{(\ell)}$ maintains the set of edges revealed in the first $\ell$ queries $q_1, ..., q_\ell$, for use in the analysis (see proof of \Cref{lemma:spectral_deterministic_lb}).}
Initialize $A^{(0)}$ to be the empty set $A^{(0)} \gets \emptyset$ \label{lineref:first}\; 
Initialize $y \in \N^{V} \gets \bm{1}$ to be the all-ones vector \; 
Initialize $z \gets (2, 1, 4, 3) \in \R^{4}$\; 
\KwOut{Responses to the queries $\{q_\ell\}_{\ell=1}^T$}
\For{$\ell \in [T]$}{
    Receive the next query $q_\ell$\; 
    \If{$q_\ell = \mathsf{GetNeighbor}(v_i^k)$}{
        $t \gets z_i$\; 
        $j \gets y[i]$ \; 
        $A^{(\ell)} \gets A^{(\ell-1)} \cup \{\{v_1^k, v_2^j\}, \{v_1^j, v_2^k\}, \{v_3^k, v_4^j\},\{v_3^j, v_4^k\}\}$ \label{lineref:second}\; 
        $y[i] \gets i + 1$\; 
        \lIf{$j \leq m$}{
            \textbf{Return:} $(m, v_t^j, 1)$ 
        }
        \lElse{
            \textbf{Return:} $\emptyset$ 
        }
    }
    \If{$q_\ell = \mathsf{GetEdge}(v_i^k, v_t^j)$}{
        \If{ $v_i^k, v_t^j \in \{\{v_1^k, v_2^j\}, \{v_1^j, v_2^k\}, \{v_3^k, v_4^j\},\{v_3^j, v_4^k\}\}$}{ 
            $A^{(\ell)} \gets A^{(\ell-1)} \cup \{\{v_1^k, v_2^j\}, \{v_1^j, v_2^k\}, \{v_3^k, v_4^j\},\{v_3^j, v_4^k\}\}$ \label{lineref:third}\;
            \textbf{Return: } True\; 
        }
        \lElse{
            \textbf{Return: } False 
        }
    }
}
\end{algorithm2e}

\begin{theorem}\label{lemma:spectral_deterministic_lb} Let $X_{\Alg}(G)$ be the output of any deterministic algorithm $\Alg$ which, given an input graph $G$ on $n$ nodes, makes at most $\frac{n^2}{128}$ queries to the generalized adjacency query model (Definition~\ref{def:query_model_gen}). Then, for any $m \geq 1$ and $n \geq 4m$, there exists an undirected graph $G$ on $n$ nodes along with a vector $v \in \R^V$, such that $\abs{v^\top \bar{L}_G v - X_{\Alg}(G)} > \frac{1}{8} v^\top v $. 
\end{theorem}
\begin{proof} Set $V = V_1 \sqcup V_2 \sqcup V_3 \sqcup V_4$ where $\abs{V_i} = m$. We use $v_i^r$ to denote the $r$-th vertex in $V_i$. Let $T \in [n^2]$, and consider the resisting oracle defined in Algorithm~\ref{alg:oracle}, and let $\tilde{E} = A^{(T)}$ (see Lines~\ref{lineref:first}, \ref{lineref:second}, and \ref{lineref:third}).

Let $E_1 = \{\{v_1^k, v_2^j\}, \{v_3^k, v_4^j\} : \{v_i^k, v_t^j\} \notin \tilde{E}, \text{ $i, t \in [4]$}\} $ and $E_2 = \{\{v_1^k, v_3^j\}, \{v_2^k, v_4^j\} : \{v_i^k, v_t^j\} \notin \tilde{E}, \text{ $i, t \in [4]$}\}$. Let $G_1 = (V, E_1 \cup \tilde{E})$ and $G_2 = (V, E_2 \cup \tilde{E})$. See Figure~\ref{fig:illustration_lb} for an illustration. Note that $G_1, G_2$ are both $m$-regular by construction. Moreover, by construction, \Cref{alg:oracle} is a valid implementation of the generalized adjacency query model (\Cref{def:query_model_gen}) on both $G_1$ and $G_2$. Hence $X_{\Alg}(G_1) = X_{\Alg}(G_2).$

Let $\overline{L}_1 = I - N_{G_1}$ be the normalized Laplacian of $G_1$ and $\overline{L}_2 = I - N_{G_2}$ be the normalized Laplacian of $G_2$. Take $v \in \R^V$ to be the vector which is 1 on $V_1 \cup V_2$ and 0 on $V_3 \cup V_4$. We have that 
\begin{align*}
    v^\top \overline{L}_i v = \frac{1}{m} \abs{\partial_{G_{i}}(V_1 \cup V_2)} \text{ and }
    v^\top v = 2m
    \text{ for }
    i \in [2]
\end{align*}
where $\partial_{G_{i}}(S) \defeq \{\{u, v\} \in E_i : u \in S, v \notin S\}$. As each query in $\ell \in [T]$ fixes at most four edges in $\tilde{E}$, we have that whenever $|{\tilde{E}}| / m < m/2$, i.e., $T < m^2/8 = n^2/128$, 
\begin{align*}
    \abs{v^\top \overline{L}_1 v - v^\top \overline{L}_2 v} \geq \frac{1}{m} \left(m^2 - |{\tilde{E}|} \right) = m - \frac{|{\tilde{E}}|}{m} > \frac{m}{2} = \frac{1}{4} v^\top v.
\end{align*}
Consequently, whenever $G \in \{G_1, G_2\}$ and the number of queries $T = \abs{Q_{\Alg}(G)} < n^2/128$, either $\abs{v^\top \overline{L}_1 v - X_{\Alg(G)}} > \frac{1}{8} v^\top v$ or $\abs{v^\top \overline{L}_2 v - X_{\Alg}(G)} > \frac{1}{8} v^\top v$. Hence, $\Alg$ must fail to output a $1/8$-additive nuclear sparsifier either when $G = G_1$ or when $G = G_2$.
\end{proof}

\addspecthm*
\begin{proof} 
We proceed by contradiction. Suppose that there is an algorithm that constructs an $\epsilon$-additive spectral sparsifier $M$ of $G$ using $o(n^2)$ Adjacency Queries. Then, we can construct an algorithm which outputs $I - M$ and achieves 
\begin{align*}
    \abs{x^\top (I- M) x - x^\top \bar{L}_G x} = \abs{x^\top x - x^\top M x - x^\top x + x^\top N_G x} = \abs{x^\top M x - x^\top N_Gx} \leq \epsilon \norm{x}_2^2.
\end{align*}
But this contradicts   Lemma~\ref{lemma:spectral_deterministic_lb} for any $\epsilon < 1/8$. As the hard instance graph in Lemma~\ref{lemma:spectral_deterministic_lb} is unweighted, the theorem holds even for unweighted graphs.
\end{proof}
\section{Sparsification in the Random Walk Model}\label{sec:rw}

In this section, we show how our notion of nuclear sparsification (and indeed, $\e$-additive spectral sparsification) can also be achieved in the one-step random walk model (Definition~\ref{def:rw_query}), which is a more restrictive query model than the one considered by \citet{BravermanKrishnanMusco:2022}. 
In \Cref{subsec:lb}, we also present a lower bound that separates the query complexity achievable for spectral sparsification from what is achievable with nuclear sparsification.

\subsection{Algorithms in the Random Walk Model}\label{subsec:algorithms}

\begin{algorithm2e}[t]
\caption{${\mathsf{NuclearSparsify}}(G, \epsilon, n, T)$}\label{alg:sparsify-nuclear}
\DontPrintSemicolon
\LinesNumbered
\textbf{Input:} Graph $G = (V=[n],E,w)$, under the one-step random walk query model (Definition~\ref{def:rw_query}), accuracy $\e \in (0,1)$,  number of queries $T \in \N$\;
\textbf{Output:} $X$, a sparsifier of $N_G$\;
Initialize $X \gets 0 \in \R^{n \times n}$\; 
\For{$t = 1, ..., T$}{ 
    $(i, j, \deg(i), \deg(j)) \gets \mathsf{RandomNeighbor}$\; 
    $X_{i,j} \gets X_{i,j} + 1/(2T) \cdot n \cdot \sqrt{\deg(i)/\deg(j)}$ \; 
    $X_{j,i} \gets X_{j,i} + 1/(2T) \cdot n \cdot \sqrt{\deg(i)/\deg(j)}$ \; 
}
\textbf{Return:} $X$
\end{algorithm2e}

We begin by presenting algorithms for sparsification under the one-step random walk query model. In Theorem~\ref{thm:nuclear_sparsifier_rw}, we show that by taking $T = \bigO(n\epsilon^{-2})$ in Algorithm~\ref{alg:sparsify-nuclear}, with constant probability, we can obtain a nuclear sparsifier. In Theorem~\ref{thm:spectral_sparsifier_rw}, we show that by taking $T = \bigO(n\epsilon^{-2}\log(n))$ in \Cref{alg:sparsify-spectral}, we can obtain an $\epsilon$-additive spectral sparsifier. 

\thmrwnuclearsparsifier*
\begin{proof} Let $i^{(t)} \sim {\mathsf{Uniform}}([n])$ and $j^{(t)} = j$ with probability proportional to $w_{i^{(t)}j}$ for each $j \in [n]$, independently for each $t \in [T]$. Let $A^{(t)} \in \R^{V \times V}$ be the matrix where 
\[
A^{(t)}_{i^{(t)},j^{(t)}} = n \cdot \sqrt{\deg(i^{(t)}) / \deg(j^{(t)})}.
\]
and all other entries are $0$. For all $i,j\in V$
\[
\E [A^{(t)}]_{i,j}
= \left[ \frac{1}{n} \cdot \frac{w_{i,j}}{\deg(i)} \right]
\cdot 
\left[ n \cdot \sqrt{\frac{\deg(i)}{\deg(j)}} \right]
= \frac{w_{i,j}}{\sqrt{\deg(i) \cdot \deg(j)}}
= N_{i,j}\,.
\]

Consequently, as $N$ is symmetric, $\E \frac{1}{2} (A^{(t)} + {A^{(t)}}^\top) = N$ and $\E \frac{1}{2T} \sum_{t \in [T]} A^{(t)} + {A^{(t)}}^\top = N$. For any i.i.d., vectors $x_1,\ldots,x_T$ $\E\norm{\frac{1}{T} \sum_{i \in [T]} x_i - \E[x_1]}_2^2 \leq \frac{1}{T} \E\norm{x_1 - \E[x_1]}_2^2 \leq \frac{1}{T} \E\norm{x_1}_2^2$. Hence, we have
\begin{align*}
    \E \norm{\paren{\frac{1}{2T} \sum_{t \in [T]} A^{(t)} + {A^{(t)}}^\top} - N }_F^2  \leq \E \paren{\norm{\frac{1}{2}  \paren{A^{(1)} + {A^{(1)}}^{\top}}}_F^2}.
\end{align*}
Furthermore,
\begin{align*}
\E \paren{\norm{\frac{1}{2}  \paren{A^{(1)} + {A^{(1)}}^{\top}}}_F^2} 
\leq \sum_{i \in [n]} \sum_{j \in [n]} \frac{w_{i,j}}{n \deg(i)} \Brac{n \cdot \sqrt{\deg(i)/\deg(j)}}^2 = n^2. 
\end{align*}
Note that $X$ in \Cref{alg:sparsify-nuclear} has the same distribution as $\frac{1}{2}  \paren{A^{(1)} + {A^{(1)}}^{\top}}$. By Markov's inequality
\begin{align*}
    \prob{ \norm{X - N}_* \geq n \epsilon} \leq \prob{ \norm{X - N}_F^2 \geq n \epsilon^2} \leq \frac{(1/T) n^2}{n \epsilon^2}.
\end{align*}
Setting $T = 3 n \epsilon^{-2}$ suffices to succeed with probability 2/3. 
\end{proof}

\begin{algorithm2e}[t]
\caption{${\mathsf{SpectralAdditiveSparsify}}(G, \epsilon, n, T)$}\label{alg:sparsify-spectral}
\DontPrintSemicolon
\LinesNumbered
\KwIn{ Graph $G = (V=[n],E,w)$, under the one-step random walk query model (Definition~\ref{def:rw_query}), Accuracy $\e \in (0, 1)$, Number of queries $T \in \N$\;}
\KwOut{$X$, a sparsifier of $N_G$\;}
Initialize $X \gets 0 \in \R^{n \times n}$\; 
\For{$t = 1, ..., T$}{ 
    $(i, j, \deg(i), \deg(j)) \gets \mathsf{RandomNeighbor}$\; 
    $Z \gets {\mathsf{Bernoulli}}(1/2)$\; 
    \If{$Z = 1$}{
        $X_{i,j} \gets X_{i,j} + \frac{2n}{\sqrt{\deg(i)} \sqrt{\deg(j)}}\paren{\frac{1}{\deg(i)} + \frac{1}{\deg(j)}}^{-1}$ \; 
    }
    \Else{
        $X_{j,i} \gets X_{j,i} + \frac{2n}{\sqrt{\deg(i)} \sqrt{\deg(j)}}\paren{\frac{1}{\deg(i)} + \frac{1}{\deg(j)}}^{-1}$ \; 
    }
}
\textbf{Return:} $X$
\end{algorithm2e}

In the following Theorem, we note that by using a result of \citet{cohen2017almost} we can also achieve the stronger notion of $\epsilon$-additive spectral sparsification in the one-step random walk access model (Definition~\ref{def:rw_query}.) 

\begin{theorem}[Theorem 3.9 (Simplified), \cite{cohen2017almost}]\label{thm:cohen}
  Let $A \in \R^{n \times n}$ be a symmetric matrix where no row or column is all zeros. Let $\e,p\in(0,1)$. Let $r = A \bf{1}$ where $\bf{1}$ is the all ones vector. Let $E_{ij}$ denote a matrix whose $(i,j)$-th entry is $1$ and rest of the entries are zero, and $\cD$ be a distribution over $\R^{n \times n}$ such that  $X \sim \cD$ takes value 
  \[ X = \paren{\frac{A_{ij}}{p_{ij}}} E_{ij} \text{ with probability } p_{ij} = \frac{A_{ij}}{2n}\Brac{ \frac{1}{r_i} + \frac{1}{r_j}} , ~\forall A_{ij} \neq 0 \mper \]
  If $A_1,\dots,A_T$ are sampled independently from $\cD$ for $T \geq 128 \frac{2n}{\e^2} \log \frac{2n}{p}$, $R = \text{Diag}(r)$ (where $\text{Diag}(r)$ is a diagonal matrix with $\text{Diag}(r)_{ii} = r_i$, for $i \in [n]$)  then the average $\tilde A \defeq \frac{1}{T} \sum_{t \in [T]} A_t$ satisfies 
  \[ \prob{ \spnorm{R^{-1/2} (\tilde A - A) R^{-1/2} } \geq \e } \leq p \mper\]
\end{theorem}

Using the above theorem, we show that with an additional multiplicative $\bigO(\log(n))$-factor in the number of queries, the \Cref{alg:sparsify-spectral} we can obtain the stronger notion of $\epsilon$-additive spectral sparsifier in the one-step random walk access model (Definition~\ref{def:rw_query}.)  

\thmrwspectralsparsifier*
\begin{proof} Let $i^{(t)} \sim {\mathsf{Uniform}}([n])$ and $j^{(t)} = j$ with probability $w_{ij}$ for each $j \in [n]$, independently for each $t \in [T]$. Let $A^{(t)} \in \R^{V \times V}$ be the matrix where with probability $1/2$, either
\[
A^{(t)}_{i^{(t)},j^{(t)}} = \frac{2n}{\sqrt{\deg(i^{(t)})} \sqrt{\deg(j^{(t)})}}\paren{\frac{1}{\deg(i^{(t)})} + \frac{1}{\deg(j^{(t)})}}^{-1},
\]
or  
\[
A^{(t)}_{j^{(t)},i^{(t)}} = \frac{2n}{\sqrt{\deg(i^{(t)})} \sqrt{\deg(j^{(t)})}}\paren{\frac{1}{\deg(i^{(t)})} + \frac{1}{\deg(j^{(t)})}}^{-1}. 
\]
All other entries are $0$. Note that then
\begin{align*}
    X_{i,j} =  \frac{2n}{\sqrt{\deg(i)} \sqrt{\deg(j)}}\paren{\frac{1}{\deg(i)} + \frac{1}{\deg(j)}}^{-1}, 
\end{align*}
with probability $\frac{1}{2n} \cdot \paren{\frac{w_{ij}}{\deg(i)} + \frac{w_{ij}}{\deg(j)}}$. By \Cref{thm:cohen}, setting $A = N_G$, for $i \in [T]$ setting $A_i = A^{(t)}$, and $p=1/3$ in Theorem~\ref{thm:cohen} and taking $T \geq 128 \frac{2n}{\epsilon^2} \log(2n/p)$, the result follows.
\end{proof}

\subsection{Separation Between Spectral and Nuclear Sparsification}\label{subsec:lb}

In this section, we show that there exists a constant $c$ such that any algorithm for constructing an $c$-additive spectral sparsifier must make at least $\Omega(n\log n)$ queries to the one-step random walk graph query oracle. Our result relies on the following classical result on the sample complexity for the coupon collector problem \cite{mitzenmacher2017probability}.

\begin{lemma}[Coupon Collector] Consider a collection of $n \geq 1$ different coupons from which coupons are drawn independently, with equal probability, and with replacement. Let $T$ be a random variable which denotes the number of trials needed to see all the $n$ different coupons. Then, for any constant $c>0$
  \[ \Pr{ T \leq n \log n - cn } \leq \exp(-c) \mper \]
\end{lemma}

In the following lemma, we obtain our lower bound for the query complexity for $\epsilon$-additive spectral sparsification by reducing to the coupon collector problem. 

\lbspectraladdcoupon*
\begin{proof}
  For each $i \in [n]$, let $X_{i}$ be an independent Bernoulli random variable. We construct a random bipartite graph on $2n$ nodes $G_X = (V,E)$ based on $X \defeq \{X_1,\dots, X_n\}$. Let $V = A \sqcup B$ denote the vertices of $G_X$, such that $\abs{A}=\abs{B}=n$. Let $a_1,\dots,a_{n}$ denote the vertices in $A$, and $b_1,\dots,b_{n}$ denote the vertices in $B$. If $X_i = 1$ we add $(a_i,b_i)$ to $E$.
  
  Let $A$ denote the adjacency matrix of $G_X$. Let $\cI \defeq \set{i : (a_i,b_i)\in E}$. The $\mathsf{RandomNeighbor}$ query either returns a vertex $a_i$ or $b_i$ for $i \in [n]\setminus \cI$, or an edge $(a_i,b_i)$ for $i \in \cI$. If the $\mathsf{RandomNeighbor}$ query returns a vertex $a_i$ (or $b_i$) for $i \not \in \cI$, then $(a_i,b_i) \not \in E$. Therefore the $\mathsf{RandomNeighbor}$ query corresponds to sampling a pair $(a_i,b_i)$ for $i \in [n]$. The pair $(a_i,b_i)$ for $i \in [n]$ corresponds to coupons in the coupon collector problem. 

  Suppose there exists an algorithm that takes  $T \leq n \log n -  \log(1/c) n$ $\mathsf{RandomNeighbor}$ queries and outputs a spectral sparsifier $\tA$ of $A$ such that $\spnorm{A-\tA} \leq \e$ with constant probability $c$. Then, we will show that we can solve the coupon collector problem in $T \leq  n \log n - \log(1/c) n$ queries with constant probability $c$, which leads to a contradiction. 
  
  This follows because we can determine all the edges in the graph $G_X$, and hence all the coupons $(a_i,b_i)$ for $i \in [n]$,  from $\tA$ by taking $x \in \R^{2n}$ indexed by $a_i$ and $b_i$, for $i \in [n]$. We take $x$ to be the all $0$ vector except at two coordinates where it is $1$, i.e., $x_{a_i} = x_{b_i} = 1$, and $x_k=0$ for $k \neq a_i 
  \text{ or } b_i$. 
  We get that if $ 3/4 \leq x^\top \tA x \leq 5/4$ then $(a_i,b_i) \in E$ and if $ -1/4 \leq x^\top \tA x \leq 2/4$, then $(a_i,b_i) \not\in E$. 
  
  Therefore, any algorithm that takes $T \leq n \log n - \log(1/c) n$ oracle queries, fails to produce a spectral sparsifier, with probability greater than $c$. 
\end{proof}

Note that for the random graphs $G_X$ considered in the proof of Theorem~\ref{lemma:lb_spectral_add_coupon}, the degree of each node is at most $1$; consequently, on these random graphs, the $k$-step random walk query model is no more informative than a $1$-step random walk model. Therefore, the lower bound of Theorem~\ref{lemma:lb_spectral_add_coupon} holds for the $k$-step query model as well. 
\section{From Nuclear Sparsifiers to Graphical Nuclear Sparsifiers}\label{app:omitted}

Throughout this paper, we have constructed nuclear sparsifiers that are matrices (Definition~\ref{def:nuclear-sparsifier-informal}) but may not necessarily be the normalized adjacency matrices of any graph. In this section, we prove the following lemma, which shows how to convert our nuclear sparsifiers into an \emph{graphical nuclear sparsifier} (i.e., a nuclear sparsifier which is the normalized adjacency matrix of some graph) using no additional queries and only $O(n/\e^2)$ additional runtime. We leverage the fact that a closer analysis of the proofs of Theorems~\ref{thm:nuclear_ub} and \ref{thm:nuclear_sparsifier_rw} actually achieve an \emph{stronger} notion of sparsification than nuclear sparsification (Definition~\ref{def:nuclear-sparsifier-informal}) in that they output $\epsilon$-additive nuclear sparsifiers $M$ such that $\norm{M - N_G}_*^2 \leq \norm{M - N_G}_F^2 \leq n\epsilon^2$.

In the following, for a matrix $A \in \R^{n \times n}$, we let $A_{[1:j]}$ denote the leading principal submatrix of order $j$.

\begin{lemma}\label{lemma:graphical-nuclear-ub} 
Let $\e \in (0, 1)$ and $G = (V, E, w)$ be a graph on $n$ nodes. Without loss of generality, assume the vertices are ordered such that $\deg(i) \geq \deg(j)$ for $i < j$. Let $M \in \R^{n \times n}$ be an matrix such that $\norm{N_G - M}_F^2 \leq n\epsilon^2$. Let $Q \defeq D_G^{1/2} M D_G^{1/2}$ and $\ebf \in \R^{n-1}$ be a vector with $e_i \defeq \deg(i) - \sum_{j=1}^{n-1} (Q)_{i,j}$. Let $G'$ be a graph with adjacency and degree matrix
\begin{align*}
A_{G'} = \begin{pmatrix}
    {Q}_{[1:n-1]} & \ebf \\
    \ebf^\top & 0
\end{pmatrix}, \text{ and } D_{G'} = \begin{pmatrix}
    {D_{G}}_{[1:n-1]}& \zero \\
    \zero^\top & \norm{\ebf}_1
\end{pmatrix}.
\end{align*}
Then, $\norm{N_{G} - N_{G'}}_* \leq 3n\epsilon$ whenever $\epsilon^{-2} \leq n$. 
\end{lemma}
\begin{proof} Recall that for any matrix $A \in \R^{n \times n}$, $\norm{A}_F \leq \norm{A}_* \leq \sqrt{n}\norm{A}_F$. Consequently, note that it suffices to show that 
\begin{align*}
    \norm{N_{G} - D_{G'}^{-1/2}A_{G'}D_{G'}^{-1/2}}_F^2 \leq 9n\epsilon^2. 
\end{align*}
By assumption, 
\begin{align*}
    \norm{\paren{N_{G} - M}_{1:n-1}}_F^2 = \norm{\paren{N_{G} - D_{G'}^{-1/2}A_{G'}D_{G'}^{-1/2}}_{1:n-1}}_F^2 \leq n \epsilon \leq n \epsilon^2.
\end{align*}
Consequently,
\begin{align}\label{eq:fix_up_degrees_frob_norm}
    \norm{N_{G} - D_{G'}^{-1/2}A_{G'}D_{G'}^{-1/2}}_F^2 &\leq n\epsilon^2 + 2\sum_{i=1}^{n-1} \paren{\frac{w_{\{i,n\}}}{\sqrt{\deg(i)} \sqrt{\deg(n)}} - \frac{e_i}{\sqrt{\deg(i)}\sqrt{\norm{\ebf}_1}}}^2\nonumber\\
    &\leq n\epsilon^2 + 2 \sum_{i=1}^{n-1} \frac{w_{\{i,n\}}^2}{{\deg(i)} {\deg(n)}} + 2 \sum_{i=1}^{n-1} \frac{e_{i}^2}{{\deg(i)} \norm{\ebf}_1}.
\end{align}
We can bound the first term in \eqref{eq:fix_up_degrees_frob_norm} by using the fact that $\deg(n) \leq \deg(i)$:
\begin{align*}
\sum_{i=1}^{n-1} \frac{w_{\{i,n\}}^2}{{\deg(i)} {\deg(n)}} \leq \frac{1}{{\deg(n)}^2}\sum_{i=1}^{n-1} w_{\{i,n\}}^2 \leq \frac{1}{{\deg(n)}^2}\left(\sum_{i=1}^{n-1} w_{\{i,n\}}\right)^2 = 1. 
\end{align*}
We can bound the second term by observing that, since since $0 \leq e_i  \leq \deg(i)$, 
\begin{align*}
\sum_{i=1}^{n-1} \frac{e_i^2}{\deg(i) \norm{\ebf}_1} \leq \frac{1}{\norm{\ebf}_1}\sum_{i=1}^{n-1} e_i = 1.
\end{align*}
Substituting into \eqref{eq:fix_up_degrees_frob_norm}, we conclude that $\norm{N_{G} - D_{G'}^{-1/2}A_{G'}D_{G'}^{-1/2}}_F^2 \leq n\epsilon^2 + 4$, which is less than $9n\epsilon^2$ whenever $\epsilon^{-2} \leq n$.  
\end{proof}

\section*{Acknowledgements}
Yujia Jin and Ishani Karmarkar were supported in part by NSF CAREER Grant CCF-1844855, NSF Grant CCF-1955039, and a PayPal research award. Yujia Jin's contributions to the project occurred while she was a graduate student at Stanford. Christopher Musco was partially supported by NSF Grant CCF-2045590. Aaron Sidford was supported in part by a Microsoft Research Faculty Fellowship, NSF CAREER Grant CCF-1844855, NSF Grant CCF-1955039, and a PayPal research award. Apoorv Vikram Singh was partially supported by NSF Grant CCF-2045590.

\appendix
\crefalias{section}{appendix}
\crefalias{subsection}{appendix}

\printbibliography

\end{document}